\documentclass[%
 notitlepage,
superscriptaddress,
nofootinbib,
 amsmath,amssymb,
 aps,
]{revtex4-1}

\usepackage[utf8]{inputenc}
\usepackage{mathpazo}
\usepackage[T1]{fontenc}
\usepackage{bbm}
\usepackage{soul}
\usepackage[normalem]{ulem}
\usepackage{braket}
\usepackage{graphicx}
\usepackage{amsmath}
\usepackage{color}

\usepackage[colorlinks,citecolor=blue,linkcolor=blue]{hyperref}

\usepackage{qcircuit}

\makeatletter
\newsavebox{\@brx}
\newcommand{\llangle}[1][]{\savebox{\@brx}{\(\m@th{#1\langle}\)}%
  \mathopen{\copy\@brx\kern-0.5\wd\@brx\usebox{\@brx}}}
\newcommand{\rrangle}[1][]{\savebox{\@brx}{\(\m@th{#1\rangle}\)}%
  \mathclose{\copy\@brx\kern-0.5\wd\@brx\usebox{\@brx}}}
\makeatother

\newtheorem{theorem}{Theorem}[section]

\newtheorem{proposition}[theorem]{Proposition}

\newtheorem{definition}{Definition}

\newenvironment{proof}[1][Proof]{\begin{trivlist}
\item[\hskip \labelsep {\bfseries #1}]}{\end{trivlist}}

\DeclareMathOperator{\tr}{tr}
\DeclareMathOperator{\id}{\mathbb{I}}

\DeclareMathOperator{\LL}{\mathcal{L}}

\newcommand{\qed}{\nobreak \ifvmode \relax \else
      \ifdim\lastskip<1.5em \hskip-\lastskip
      \hskip1.5em plus0em minus0.5em \fi \nobreak
      \vrule height0.75em width0.5em depth0.25em\fi}

\begin{document}
\title{Observer-dependent locality of quantum events}
\author{Philippe Allard Gu\'{e}rin}
\affiliation{Faculty of Physics, University of Vienna, Boltzmanngasse 5, 1090 Vienna, Austria}
\affiliation{Institute for Quantum Optics and Quantum Information (IQOQI), Austrian Academy of Sciences, Boltzmanngasse 3, 1090 Vienna, Austria}

\author{\v{C}aslav Brukner}
\affiliation{Faculty of Physics, University of Vienna, Boltzmanngasse 5, 1090 Vienna, Austria}
\affiliation{Institute for Quantum Optics and Quantum Information (IQOQI), Austrian Academy of Sciences, Boltzmanngasse 3, 1090 Vienna, Austria}

\date{\today}

\begin{abstract}

In general relativity, the causal structure between events is dynamical, but it is definite and observer-independent; events are point-like and the membership of an event $A$ in the future or past light-cone of an event $B$ is an observer-independent statement. When events are defined with respect to quantum systems however, nothing guarantees that the causal relationship between $A$ and $B$ is definite. We propose to associate a \textit{causal reference frame} corresponding to each event, which can be interpreted as an observer-dependent time according to which an observer describes the evolution of quantum systems. In the causal reference frame of one event, this particular event is always localised, but other events can be "smeared out" in the future and in the past. We do not impose a predefined causal order between the events, but only require that descriptions from different reference frames obey a global consistency condition. We show that our new formalism is equivalent to the pure process matrix formalism~\cite{Araujo2017}. The latter is known to predict certain multipartite correlations, which are incompatible with the assumption of a causal ordering of the events -- these correlations violate causal inequalities. We show how the causal reference frame description can be used to gain insight into the question of realisability of such strongly non-causal processes in laboratory experiments. As another application, we use causal reference frames to revisit a thought experiment~\cite{Zych2017} where the gravitational time dilation due to a massive object in a quantum superposition of positions leads to a superposition of the causal ordering of two events.

\end{abstract}

\maketitle

\section{Introduction}

The usual formalism of quantum mechanics explicitly depends on a background spacetime; this is indeed one of the major conceptual obstacles to a quantum theory of gravity~\cite{Rovelli2004, Smolin2011, Kiefer2012, Isham1993}. In quantum field theory, we must first specify a space-time with a fixed metric, before we can define quantum fields as operator-valued distributions on this space-time. Matter is described by quantum mechanics, and it is allowed to be in a quantum superposition of two positions. But since Einstein's equations relates the mass-energy distribution to the metric, we expect something like a "quantum superposition of spacetime metrics" to accompany the superposition of position of the matter~\cite{Feynman1957}. As remarked by Butterfield and Isham, \textit{once we embark on constructing a quantum theory of gravity, we expect some sort of quantum fluctuations in the metric, and so also in the causal structure. But in that case, how are we to formulate a quantum theory with a fluctuating causal structure?}~\cite{Butterfield1999} 

Regardless of the specific details of an underlying theory of quantum gravity, the superposition principle makes it reasonable to expect that in some low-energy limit (whose precise nature could only be rigorously established from a complete theory) quantum superpositions of classical solutions to Einstein's equations can occur. In the recent work by Zych et al.~\cite{Zych2017}, it is argued that a quantum superposition of matter could lead to the quantum superposition of the causal orders of two events~\cite{Chiribella2012}, due to gravitational time-dilation. Their description of the situation proceeds from the point of view of a far-away observer who is not affected by the gravitational field. One might question whether such an outside description is necessary, and ask whether indefinitely-causal processes admit a \textit{relational} description~\cite{Rovelli1996}, from the point of view of the local observers.

In general relativity, events are defined with respect to localised physical systems (for example, the intersection of the world-lines of two particles is an event), and causality is a relationship between events. A classical event is "point-like": mathematically it is represented by an equivalence class, with respect to the diffeomorphism group, of points on the spacetime manifold. Can a similar definition of events be provided for quantum systems, and if so, will the point-like nature of events persist?

In this work we provide an operational definition of events for quantum systems, and study causality as the relationship between such events. We formalise this in Section~\ref{sec:causal_frames}, where we associate an observer to each event, and postulate a corresponding \textit{causal reference frame}, which may be interpreted as an observer-dependent time that the observer uses to parametrise the evolution of quantum systems. We do not preimpose a well-defined global ordering of the events; we tolerate that according to one event's causal reference frame, the other events might not necessarily be localised in the future or in the past. Instead, we require a weaker \textit{consistency condition}: all observers should agree about the evolution connecting the state in the distant past to the state in the distant future. Thus, the observer-independent localisation of events in general relativity -- the fact that events can be modelled as points on a space-time manifold that is common to all observers --  is weakened, but the consistency condition guarantees that the global causal structure (mathematically, the corresponding process matrix) is still observer-independent.

There is a concrete need to understand how the usual ideas of causality (which depend on a fixed classical metric) are modified by quantum mechanics; formalisms that may help address this question have been proposed by Hardy~\cite{Hardy2005, Hardy2007, Hardy2009, Hardy2016} and Oeckl~\cite{Oeckl2008,Oeckl2013, Oeckl2016}. The process matrix formalism~\cite{Oreshkov2012} is closely related to the above approaches, and also makes it possible to study multipartite quantum correlations without the assumption of a definite causal order between the parties. The quantum switch~\cite{Chiribella2013} is an example of a non-causal process that has been implemented in the laboratory~\cite{Procopio2015, Rubino2017}. Other processes can violate device-independent causal inequalities; unfortunately these processes are so far lacking a physical interpretation. Non-causal processes offer interesting advantages for information processing~\cite{Chiribella2012, Araujo2014_PRL, Feix2015, Guerin2016, Baumeler2017, Baumeler2018, Araujo2017_CTC}, so it is important to understand which of them could in principle be implemented in the laboratory. Recently, Ara\'{u}jo et al.~\cite{Araujo2017} have defined pure processes (that can be understood as unitary supermaps) and proposed a purification postulate, which rules out processes that do not admit a purification. One of their motivations for imposing this requirement is that only purifiable processes are compatible with the cherished reversibility of the fundamental laws of physics, in that they do not cause "information paradoxes". Nonetheless, some pure processes are known to violate causal inequalities, showing that purifiability alone is not enough to single out the processes with a known physical implementation.

In Section~\ref{sec:pure_formalism}, we show that there is a one-to-one correspondence between pure process matrices and our new description of quantum causal structures in terms of causal reference frames. This equivalence yields a different physical justification for the purification postulate of Ref.~\cite{Araujo2017}: pure processes are those that allow an (observer-dependent) description in terms of a quantum system evolving in time. 
We show how known examples of processes can be understood in terms of causal reference frames. Causally ordered processes are those for which the locality of events is observer independent: in the causal reference frame of any event, all other events are localised either in the past or in the future. A more interesting example is the quantum switch, where according to event $A$'s causal reference frame, event $B$ is in a controlled superposition of being in the future or in the past (and vice-versa). In Section~\ref{sec:swiss}, we study a new example of a causal inequality violating pure tripartite process, obtained by taking the time-reverse of a known non-causal classical process~\cite{Baumeler2015, Araujo2017}. We point out some curious features in the causal frame description of this process, which may explain why such processes do not have a known realisation in the laboratory.

Finally, in Section~\ref{sec:gravity_switch}, we revisit the thought experiment of the gravitational quantum switch, and show how causal reference frames can be applied in that context. We show how a judicious change of coordinates can be used to bring two different classical spacetimes into a form that corresponds to the causal reference frame of a particular event. We then invoke the superposition principle and obtain a representation of the gravitational quantum switch, in the causal reference frame of that event.

The consequences of the fact that quantum mechanical events do not occur at a well-defined instant in time has been studied by many authors, including among others Refs.~\cite{Peres1985, Reisenberger2002, Oppenheim2000, Brunetti2002, Micanek1996, Aharonov1998}, and the lack of  a ``common time reference'' shared between the parties in the quantum switch was noted in Ref.~\cite{Oreshkov2016}. Motivated by the question of whether experimental implementations of the quantum switch can be considered to be genuine, Oreshkov recently argued that in bipartite pure processes, the parties can be said to act on ``time-delocalised subsystems''~\cite{Oreshkov2018}. Our approach is complementary and seeks to describe the \textit{time evolution} of a quantum system according to a reference frame associated to one of the parties. In Section~\ref{sec:oreshkov} we comment on the mathematical link between the two approaches, and answer in the affirmative to a question that was raised in Ref.~\cite{Oreshkov2018} concerning the existence of a specific representation for all pure multipartite processes.

\section{Quantum theory in the frame of a localised observer}
\label{sec:causal_frames}

\subsection{Events and causality}
\label{subsec:events}

According to Wald's influential textbook on general relativity~\cite{Wald1984}, \textit{we can consider space and time ($\equiv$ spacetime) to be a continuum composed of events, where each event can be thought of as a point of space at an instant of time.} The diffeomorphism invariance of general relativity brings difficulties to the view that points in the spacetime manifold have a physical meaning, via the famous hole argument. If one wants to give physical meaning to the points in the spacetime manifold, then one must conclude that the dynamics of general relativity is underdetermined: a set of initial conditions for the gravity and matter fields (for example on a space-like hypersurface) does not uniquely determine the values for the fields at other points of spacetime, due to the gauge-symmetry corresponding to diffeomorphism invariance. We refer the reader to the reviews~\cite{Norton_SEP, Stachel2014} and references therein for a detailed treatment of the hole arguments and its implications.

Instead, events can be meaningfully defined with respect to physical systems. For example, it might be possible to identify an event unambiguously via statements such as "the place in spacetime where a particular clock reads 10 o'clock", or "the place in spacetime where these two billiard balls collide with each other". More generally, some authors (going all the way back to Einstein) have defined events operationally and in a diffeomorphism-invariant manner via the coincidences of fields or worldlines: an incomplete selection of such approaches is Refs.~\cite{Westman2008, Rovelli1991, Bergmann1961, Hardy2016}. After the physical identification of an event is made, the locality of an event -- its point-like nature -- is observer independent: spatiotemporally distant observers will also agree that the event happened at a some spacetime point. 

Once a set events has been identified with respect to physical systems, we can start asking about the causal ordering between those events. This is done by postulating observers \textit{at these events} that may (or may not) signal to each other by manipulating physical systems. The inclusion of observers (and the ``free choice'' assumption for some of their actions) allows us to characterise a causal structure by the possibilities it offers for signalling. In relativity, the analysis of the causal relations between events is conceptually straightforward and is dictated by the spacetime metric. An event $B$ is contained in the causal future of an event $A$ if there exists a future-directed spacetime path connecting $A$ to $B$ such that the tangent vector to this path is everywhere time-like or null. If this is the case, then $A$ can signal to $B$ (an intervention at $A$ can in principle affect the probabilities for observations at $B$); otherwise signalling  from $A$ to $B$ is impossible. Interestingly, the information about the causal structure between all points of a spacetime is sufficient to reconstruct the topology of spacetime, as well as the metric up to a conformal transformation~\cite{Malament1977}. This fact suggests that studying the causal structures allowed by quantum mechanics could lead to insights about the nature of spacetime in the quantum regime.

Given that our most fundamental theory of matter is quantum mechanics, it is natural to wish for a definition of "quantum events" in terms of quantum systems. There are difficulties in following the strategy of general relativity, and defining spacetime via coincidences of quantum fields, because the matter fields generally don't take well-defined values and are instead represented by operators on a Hilbert space. But one also cannot rely on the existence of classical "rods and clocks" to define events: as a concrete example, consider spontaneous emission -- the phenomena by which the interaction of an atom with the quantised electromagnetic field makes the excited states of the atom decay. Suppose that an atom is prepared in an excited state, while the electromagnetic field is in the vacuum state. To this preparation is associated a probability for a photon to be detected by a nearby detector, after a certain amount of external time has elapsed. The detection of the photon by the detector defines an event, but this event does not occur at a pre-defined value of the background time. This simple example motivates our requirement that quantum mechanical events (such as the "clicking" of a detector) must be identified with respect to physical systems, rather than by referring to an external classical space-time.

In this work, we take an operational approach that does not refer to a background spacetime, and identify events with basic experimental procedures that act on quantum systems. More precisely, an event is operationally defined with respect to a localised physical system (we refer to the region in which this system is localised as a ``laboratory'') and consists in four ``instantaneous'' steps
\begin{enumerate}
\item Heralding: a signal asserts that the system has entered the laboratory.~\footnote{The heralding step is actually quite subtle to treat in full generality. The signal could come from a measurement on the incoming quantum system, from a classical clock inside the laboratory (waiting until a specific time at which the system is guaranteed to arrive), etc. We will assume that the heralding does not disturb the state of the system and that it happens with probability one (the system is guaranteed to eventually enter the lab). This last assumption is a restriction on the generality of the approach: there are physically relevant situations where the system only enters the lab with some probability, and it would be interesting to further investigate how one can treat such probabilistic events.}
\item Intervention: a choice $x$ is made for the operation to be applied on the system
\item Observation: A classical outcome $a$ is recorded. 
\item Output: a physical system (whose state generally depends on $x$ and $a$) exits the laboratory.
\end{enumerate}

\begin{figure}[t]
\includegraphics[width=3cm]{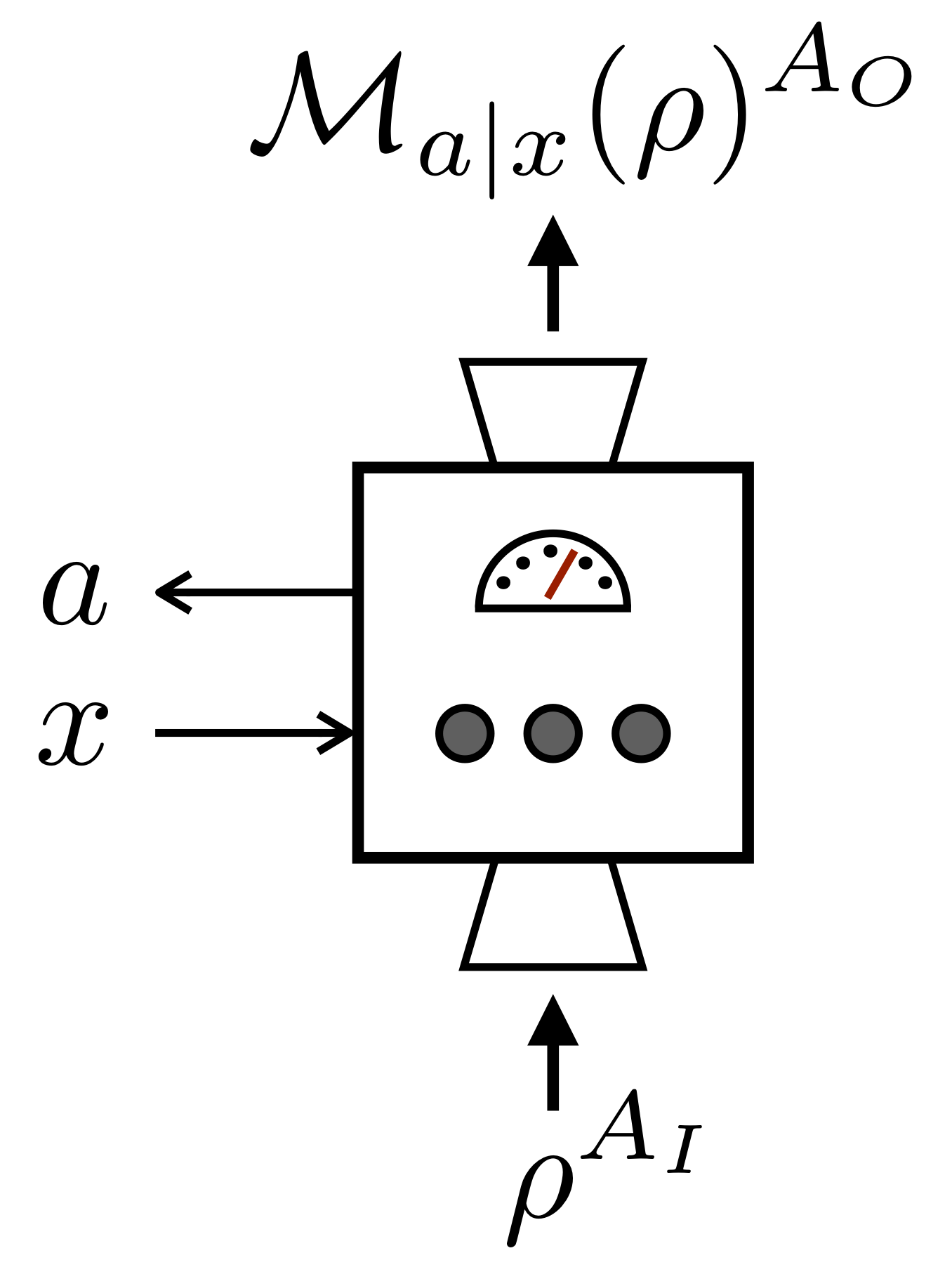}
\caption{The input system is described by density operator $\rho$ on the Hilbert space $A_I$, and the entering of the system in the laboratory heralds the event. After a choice of setting $x$ and the recording of an outcome $a$, the (unnormalised) state of the output system is $\mathcal{M}_{a|x}(\rho)$.}
\centering
\label{fig:quantum_event}
\end{figure}

In quantum theory on finite-dimensional Hilbert spaces, we associate a Hilbert space $A_I$ to the input system of the laboratory, and $A_O$ to its output system. Figure~\ref{fig:quantum_event} depicts the intervention, observation, and output steps. In the intervention step, a choice $x$ of a quantum instrument that acts on the input quantum system is made: this is represented by a collection $\{\mathcal{M}_{a|x}: \mathcal{L}(A_I) \to \mathcal{L}(A_O) \}$ of completely positive maps such that outcome $a$ occurs with probability $p(a|x) = \mathrm{tr}\left( \mathcal{M}_{a|x}( \rho) \right)$, when the input state is $\rho$. The output state is then
\begin{equation}
\rho^{A_I} \mapsto \frac{\mathcal{M}_{a|x}(\rho^{A_I})}{ \mathrm{tr}\left( \mathcal{M}_{a|x}( \rho) \right)}.
\label{eq:inst_inside}
\end{equation}
The normalisation of probabilities enforces that $\sum_a \mathcal{M}_{a|x}$ is a completely-positive trace-preserving (CPTP) map for all $x$.

\subsection{The causal reference frame of an event}
\label{subsec:causal_ref}

In this work, we study causality by using operationally defined events as our basic ingredients, rather than relying on a background space and time. We associate to each event an observer, from whose point of view we may describe physics; we will sometimes use the term \textit{observer-event} to emphasise this. To make a connection with the usual time-ordered descriptions of physics, we postulate that there is a \textit{causal reference frame} (which can be interpreted as an observer-dependent time function) associated to each observer-event, according to which this event is localised in space and in time. According to this reference frame, there should be a well-defined evolution from the past to the present, and from the present to the future. 

As discussed in the previous section, events are defined with respect to a quantum system. We will consider events that are defined with respect to a subsystem of some "global" quantum system whose Hilbert space is $\mathcal{H}$. This Hilbert space can be decomposed as $\mathcal{H} = A_I \otimes E_A$, where $A_I$ is identified with the input space of Alice's laboratory, and where $E_A$ is an "environment" on which Alice acts as the identity. The output Hilbert space of Alice's laboratory is $A_O$ and it is assumed to be isomorphic to $A_I$.~\footnote{It is not a real restriction to impose that $A_I \cong A_O$, because we can emulate any process that has $\mathrm{dim} A_I \neq \mathrm{dim} A_O$ by enlarging the smallest Hilbert space and tracing over the unwanted dimensions. It is convenient to label the Hilbert spaces of the global system in the past and in the future -- both isomorphic to $\mathcal{H}$ -- differently, as $P$ and $F$.} If we assume that the global quantum system $\mathcal{H}$ is isolated at times other than that of the event, then the evolutions from $P$ to $A_I \otimes E_A$, and from $A_O \otimes E_A$ to $F$ are unitary. Thus there exists unitaries $\Pi_A: P \to A_I \otimes E_A$ and $\Phi_A: A_O \otimes E_A \to F$ such that the global evolution from past to future, when Alice performs the quantum instrument $\{\mathcal{M}_{a|x} \}$ is
\begin{equation}
\phi_A \circ (\mathcal{M}_{a|x}^{A_I A_O} \otimes \mathcal{I}^{E_A} ) \circ \pi_A,
\label{eq:light-cone_time}
\end{equation}
where $\pi_A(\rho) = \Pi_A \rho \Pi_A^\dagger$, $\phi_A(\rho) = \Phi_A \rho \Phi_A^\dagger$, and $\mathcal{I}^{E_A}$ is the identity map on the environment degrees of freedom. If the state in the distant past is $\rho^P$, the outcome $a$ of the instrument occurs with probability 
\begin{equation}
p(a|x) = \mathrm{tr}\left( \mathcal{M}_{a|x}^{A_I A_O} (\mathrm{tr}_{E_A}(\Pi_A \rho^P \Pi_A^\dagger)\right).
\end{equation}
Equivalently, this evolution can be represented with a quantum circuit as
\begin{equation}
\includegraphics[width=6cm]{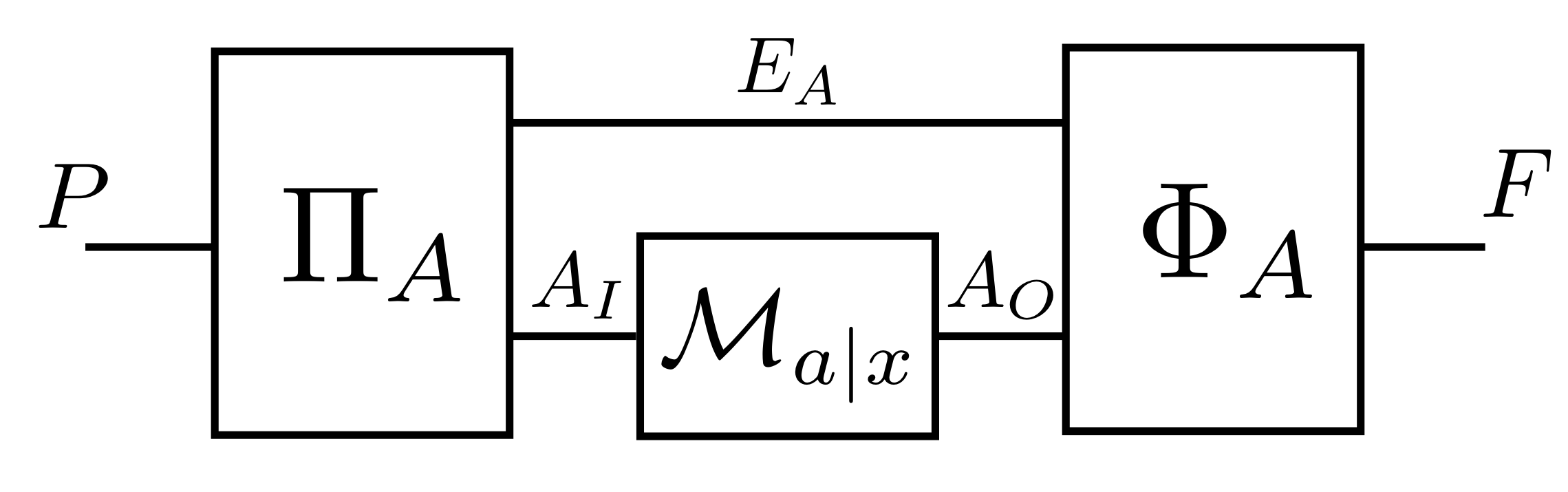}
\label{causal_frame_circuit}
\end{equation}

There is some arbitrariness in the decomposition of the evolution into a past $\Pi_A$ and a future $\Phi_A$, which is similar to the arbitrariness in choosing a time coordinate in relativity. For example, the parts of the evolution that are at "space-like separated locations" from Alice's event can be arbitrarily moved to the future or to the past. Our definition of causal reference frames will deal with this arbitrariness.

Consider now a physical situation comprising of more than one observer-event; for simplicity of notation we consider only two of them, Alice and Bob, whose respective input and output Hilbert spaces are assumed to be isomorphic: $A_I \cong A_O, B_I \cong B_O$. As in Section~\ref{subsec:causal_ref}, each party has an associated causal reference frame, with corresponding unitaries $\Pi_A, \Phi_A$ and $\Pi_B, \Phi_B$. We make the assumption of ``free-choice'', to guarantee that the choice of operation made by the parties can be treated as an independent classical variable. Thus we treat the unitary evolutions in the past and future of Alice's event, $\Pi_A,\Phi_A$ as functions of Bob's choice of instrument, and vice-versa. For the time being, we consider only the case where both parties are performing unitaries on their quantum system, in which case we make the following definitions.
\begin{definition}(Frame functions)

A frame function for Alice is a pair of functions $(\Pi_A, \Phi_A)$ each sending unitaries $U_B: B_I \to B_O$ to unitaries
\begin{align}
\Pi_A(U_B) &: P \to A_I \otimes E_A\\
\Phi_A(U_B) &: A_O \otimes E_A\to F.
\end{align}
\end{definition}
Note that it might not always be possible to extend the domain of the functions $\Pi_A, \Phi_A$ to include all linear operators.~\footnote{We thank David Trillo Fernandez for this observation}

\begin{definition}(Causal reference frame)

Alice's causal reference frame is an equivalence class of frame functions. Two frame functions $(\Pi_A, \Phi_A)$ and $(\Pi_A', \Phi_A')$ are equivalent if
\begin{equation}
\Phi_A(U_B) (U_A \otimes \id^{E_A}) \Pi_A(U_B) = \Phi_A'(U_B) (U_A \otimes \id^{E_A}) \Pi_A'(U_B)
\end{equation}
for all unitaries $U_A:A_I \to A_O$ and $U_B: B_I \to B_O$. In the above, $\id^{E_A}$ is the identity operator that acts on $E_A$. Bob's causal reference frame is defined analogously, with the obvious modifications.
\end{definition}
There are in general many frame functions in the equivalence class. For example, given a frame function $(\Pi_A, \Phi_A)$, and an arbitrary unitary $V^{E_A}$, we see that $\Phi_A' = \Phi_A (\id^{A} \otimes V^{E_A})$ and $\Pi_A' = (\id^A \otimes (V^\dagger)^{E_A}) \Pi_A$ belong to the same causal reference frame. We will usually use a single frame function $(\Pi_A, \Phi_A)$ to define a causal reference frame, where implicitly we mean that $(\Pi_A, \Phi_A)$ is one representative of the equivalence class.

When there are more than one party involved in the process, each party will have its associated causal reference frame. We want to formulate, in the least restrictive way as possible, the requirement that the causal reference of both parties are describing the \textit{same} physical process. In every observer's event causal reference frame (according to its "observer-dependent time function") there is a time in the distant past before which none of the parties has acted yet, and a time in the future after which all the parties have finished acting. We impose a \textit{consistency requirement} (defined formally in Def.~\ref{def:process_rel}), to ensure that the unitary mapping "in" states to "out" states at these distant times is the same for all observers, but we will not assume a well-defined ordering of the events. The role of this requirement is to enforce that the parties are describing the same physical situation.\footnote{It should not be interpreted as equivalent to logical consistency, in the same way that a situation in which Alice says "hello" and Bob hears "goodbye" is logically consistent, but relatively uninteresting for physics.} A way to interpret this requirement is that we want the global evolution from $P$ to $F$ to be observer independent, but we allow its decomposition into a ``past'' and a ``future'' to depend on the observer.
 
 \begin{definition}(Consistent causal reference frames)
 \label{def:process_rel}
 
A pair of causal reference frames $(\Pi_A, \Phi_A), (\Pi_B, \Phi_B)$ for Alice and Bob are \textit{consistent}\footnote{One might object to the fact that the consistency condition of Eq.~(\ref{eq:consistency}) supposes that the parties are describing the state at $P$ and $F$ in the same basis. We could have also defined the consistency condition "up to unitary": in that case, the consistency condition would be that there exists constant unitaries $U, V$ such that $\Phi_A(U_B) (U_A \otimes I^{E_A}) \Pi_A(U_B) = U \Phi_B(U_A) (U_B \otimes I^{E_B}) \Pi_B(U_A) V$. However, this change of basis does not change anything for causality, and can be dealt with separately.} if for all unitaries $U_A: A_I \to A_O$ and $U_B: B_I \to B_O$,
\begin{equation}
\Phi_A(U_B) (U_A \otimes \id^{E_A}) \Pi_A(U_B) = \Phi_B(U_A) (U_B \otimes \id^{E_B}) \Pi_B(U_A):= \mathcal{G}(U_A,U_B).
\label{eq:consistency}
 \end{equation}
\end{definition}
Equivalently, in circuit notation, the consistency condition means that the frame functions must satisfy
\begin{equation}
\includegraphics[width=12cm]{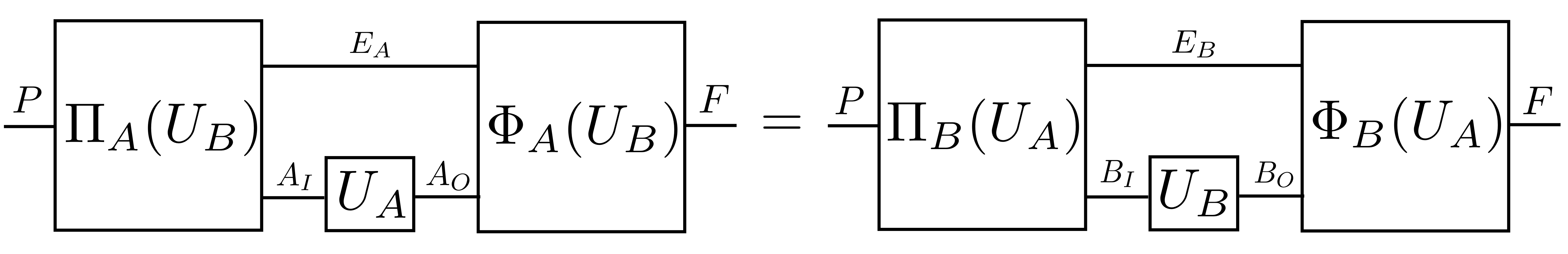},
\label{consistency_circuit}
\end{equation}
for all unitaries $U_A, U_B$.
 
One might prefer the consistency requirement to be formulated purely in terms of device-independent quantities, such as probabilities for the outcomes of measurements. This can be achieved without changing the mathematical description, simply by reinterpreting $P$ as the output Hilbert space of a third party, and $F$ as the input Hilbert space of a fourth party. For the operationally inclined, the quantum evolution $\mathcal{G}(U_A, U_B)$ between $P$ and $F$ is just a concise encoding for the probabilities of measurement outcomes at $F$, conditional on a state preparation at $P$, and on the applied unitaries $U_A, U_B$.

Definition~\ref{def:process_rel} can be easily generalised to $N$ parties $A_1, A_2, ... A_N$, in which case $\Pi_{A_1}$ will be a function of $U_{A_2}, ..., U_{A_N}$, etc. The consistency condition is then to be imposed between the causal frames of all parties.

Our definitions do not yet specify what happens when the parties perform general quantum instruments. In order to study phenomena such as the violation of causal inequalities, we need to calculate the outcome probabilities of general quantum instruments, and in the case the evolution from $P$ to $F$ will be a general quantum channel. Fortunately, we will show in Section~\ref{sec:pure_formalism} that formulating our definitions uniquely in terms of unitaries is not a restriction. Indeed, we will prove that Eq.~(\ref{eq:consistency}) for the action of $\mathcal{G}$ on unitaries uniquely specifies a pure process matrix~\cite{Araujo2017}, which can then be used to calculate the outcome probabilities for general quantum instruments. Said differently: if we want to extend $\mathcal{G}$ to a linear map on quantum instruments that agrees with Eq.~(\ref{eq:consistency}) for unitaries, there is a unique way to do so.

However, before we turn to proving the equivalence with the process matrix formalism, we give a few examples of processes that admit a description in terms of consistent causal reference frames.

\subsection{Example: causally ordered process}

A causally ordered bipartite process is one in which one of the parties cannot signal to the other. A process has the order $A \leq B$ if no matter his choice of local operation, $B$ cannot signal to $A$. In general, all pure bipartite processes with causal order $A \leq B$ are "channels with memory", of the form~\cite{Chiribella2009}
\begin{equation}
\Qcircuit @C=1em @R=1em @!R {
&\multigate{1}{V_1}& \qw  & \multigate{1}{V_2}& \qw & \multigate{1}{V_3}&\qw \\
& \ghost{V_1}& \gate{U_A} & \ghost{V_2} & \gate{U_B} & \ghost{V_3}&\qw & 
},
\end{equation}
 for some fixed unitaries $V_1, V_2, V_3$. We see directly that the above circuit can be used to represent both Alice's causal frame and Bob's causal frame. Therefore, for causally ordered processes it is possible to find a causal reference frame in which both $A$ and $B$ are both "localised in time": in the above we have that $B$ is localised in the future of $A$.
 
\subsection{Example: the quantum switch}

An interesting example of a physically relevant process that does not possess a well-defined causal order is the quantum switch~\cite{Chiribella2013, Araujo2015}. Nevertheless, we can choose any single observer, and decompose the process into a past and a future relative to his observer-event. Furthermore, it is possible to describe the past and future evolutions in a unitary way. The simplest version of the quantum switch is a bipartite process with $\mathrm{dim}( P) = \mathrm{dim}( F) = 4$ and $\mathrm{dim}( A) = \mathrm{dim}( B) = 2$. In circuit notation, we can write it according to Alice's causal reference frame as
\begin{equation}
\label{eq:switch_A}
\Qcircuit @C=.5em @R=0em @!R {
& \ctrl{1} &  \qw &\ctrlo{1} & \qw \\
&\gate{U_B} & \gate{U_A} & \gate{U_B}& \qw & , 
}
\end{equation}
while in Bob's causal reference frame it is
\begin{equation}
\label{eq:switch_B}
\Qcircuit @C=.5em @R=0em @!R {
& \ctrlo{1} &  \qw &\ctrl{1} & \qw \\
&\gate{U_A} & \gate{U_B} & \gate{U_A}& \qw & .
}
\end{equation}
In the above circuits, the upper qubit is the control-qubit, denoted by $C$, and the lower qubit is the ``target'' qubit which we denote by $S$. A black circle means control on the state $|1 \rangle_C$, while a white circle is a control on the state $|0 \rangle_C$. It is straightforward to check that both circuits yield the same global evolution $\mathcal{G}(U_A, U_B)$ from $P$ to $F$. This example shows that the consistency condition can be satisfied by processes in which one of the parties is delocalised in time: here we have $\Pi_A(U_B) = |0 \rangle \langle 0|^C \otimes \id^S + |1\rangle \langle 1|^C \otimes U_B^S$ and $\Phi_A(U_B) = |0 \rangle \langle 0|_C \otimes U_B + |1 \rangle \langle 1|^C \otimes \id^S$.

A common argument (see the Supplementary information of Ref.~\cite{Maclean2017}) attempts to conclude that the quantum switch, as realised in quantum optics experiments is "not the real thing", in that it can be described with a spacetime diagram that involves two space-time points per party, rather than only one. However as discussed in Section~\ref{subsec:events}, space-time points do not have an priori physical meaning even in classical physics, and one should not expect them to fare better once quantum mechanics enters the picture. The time-delocalisation of a local operation in the quantum switch does not mean that the operation is performed multiple times; it is executed only once, but on a \textit{time-delocalised subsystem}, as argued by Oreshkov~\cite{Oreshkov2018}. Our approach with causal reference frames provides a means to describe any pure process as the observer-dependent time evolution of a quantum system; during this evolution the time-localisation of events is generally observer dependent as shown by Eqs.~(\ref{eq:switch_A},~\ref{eq:switch_B}) in the case of the quantum switch.

\section{Equivalence with the process matrix formalism}
\label{sec:pure_formalism}

In Definition~\ref{def:process_rel}, we have proposed a relational definition of "processes" as a set of causal reference frames that obey a consistency condition. In this section, we make an explicit connection between the already existing process matrix formalism~\cite{Oreshkov2012} and the newly developed language of causal reference frames. Namely, we show that pure processes~\cite{Araujo2017} are in one-to-one correspondence with consistent causal reference frames. This equivalence will also show that we were justified, in the previous section, in limiting our definitions to the unitary case. The notation for the process matrix formalism relies heavily on the channel-state duality, or Choi-Jamio\l kowski isomorphism, which is reviewed in Appendix~\ref{app:CJ}. In the following, we follow common usage in the literature, where the terms ``process'' and ``process matrix'' are used interchangeably (altough the latter could be seen as the mathematical representation of the former; this is analogous to the relation between the terms ``quantum state `` and ``density matrix'').

\subsection{Pure processes}

In the original paper by Oreshkov, Costa and Brukner~\cite{Oreshkov2012}, a process matrix is defined as a functional on quantum instruments, obeying the requirement that probabilities are well-defined for all possible operations of the parties, including operations that involve shared entangled ancillary systems (this last condition ensures that the process matrix is positive semidefinite). It can be more convenient to view process matrices as ``supermaps''~\cite{Chiribella2008} that takes the local quantum channels of the parties and sends them to a quantum channel from a past Hilbert space $P$ to a future Hilbert space $F$. General formalisms for higher-order transformations, which include process matrices as special cases are presented in Refs.~\cite{Perinotti2017, Kissinger2017,Bisio2018}, and it would be interesting to investigate whether an analogous theory of causal reference frames can be developed for these more general frameworks. 

\begin{figure}[t]
\includegraphics[width=6cm]{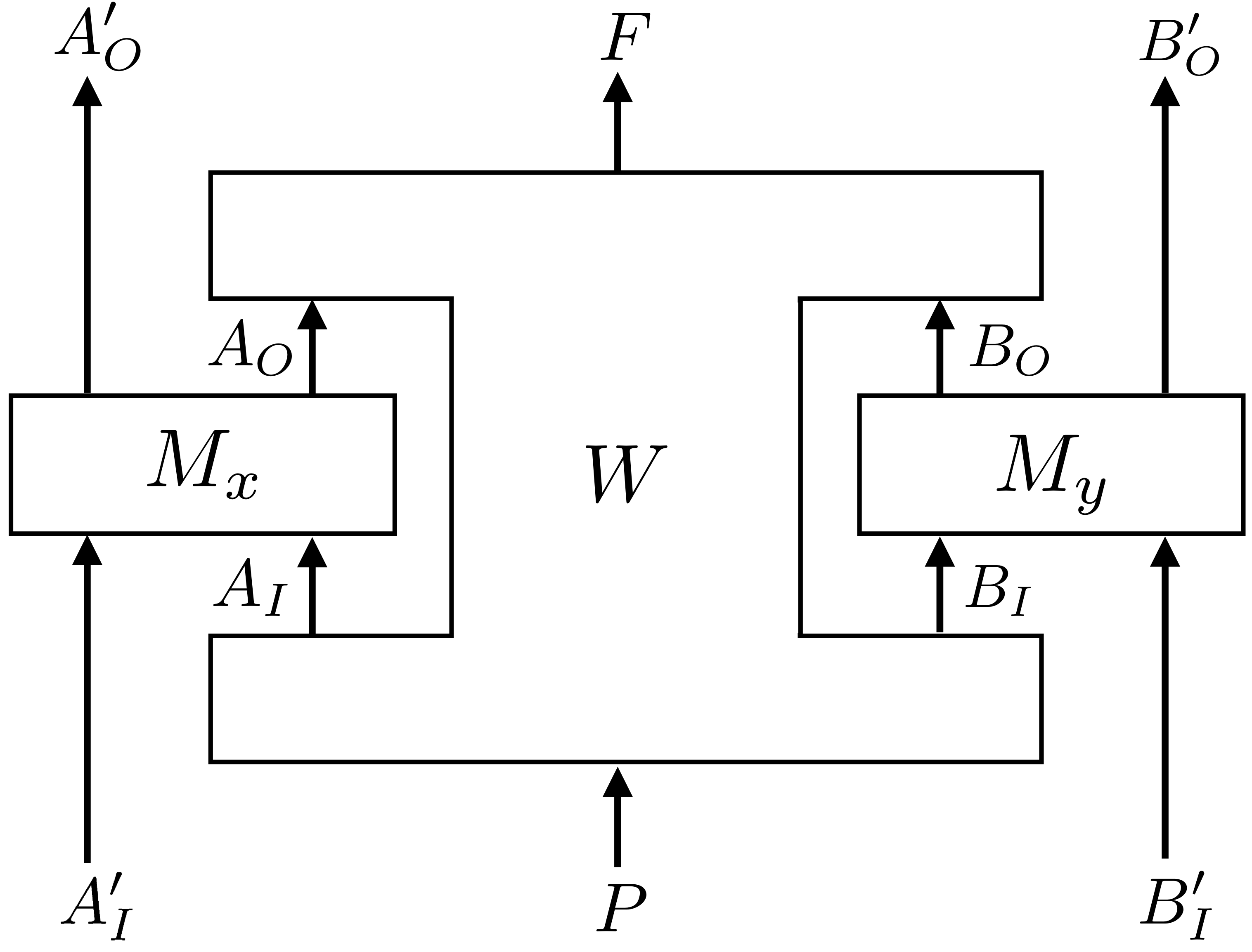}
\caption{Graphical representation of the relevant Hilbert spaces from Def.~\ref{def:process} of a bipartite process matrix.}
\centering
\label{fig:process}
\end{figure}

For the sake of simplicity, in what follows we consider only two parties, Alice and Bob, as shown in Figure~\ref{fig:process}. The extension of the definitions to more parties is straightforward, and all the results of this section continue to hold in the multipartite case. The localised laboratory of Alice has a finite dimensional input Hilbert space $A_I$ and output space $A_O$; similarly Bob has input $B_I$ and output $B_O$. We further allow the parties to have arbitrary ancillary Hilbert spaces $A_I', A_O', B_I', B_O'$, which are directly connected to the future (resp. past), as shown in Figure ~\ref{fig:process}. A quantum channel for Alice is a completely positive and trace-preserving (CPTP) map $\mathcal{M}: \mathcal{L}(A_I A_I') \to \mathcal{L} (A_O  A_O')$, where tensor products are implied so that $A_I A_I' = A_I \otimes A_I'$. Equivalently, the Choi state of a CPTP map (see the review in Appendix~\ref{app:CJ}) obeys $M^{A_I A_I' A_O A_O'} \geq 0$ and $\tr_{A_O A_O'} M^{A_I A_I' A_O A_O'} = \id^{A_I A_I'}$. We sometimes use superscripts to indicate the Hilbert spaces on which an operator acts.

We define process matrices as in Ref.~\cite{Araujo2017},~\footnote{However, our notation differs in that we use $A_I$ for Alice's input Hilbert space, while Ref.~\cite{Araujo2017} uses $A_I$ for the space of matrices acting on the input Hilbert space.}

\begin{definition}(Process matrix)
\label{def:process}
An operator $W^{P F A_I A_O B_I B_O} \in \mathcal{L}(P  F  A_I  A_O B_I B_O)$ is a process matrix if for all CPTP maps $\mathcal{M}_x: \mathcal{L}(A_I  A_I' )\to \mathcal{L}(A_O A_O')$, $\mathcal{M}_y: \mathcal{L}(B_I B_I' )\to \mathcal{L}(B_O B_O')$, where $A_I', A_O', B_I', B_O'$ are ancillary Hilbert spaces of arbitrary dimension, the operator
\begin{equation}
\label{eq:resulting_G}
G_{xy} = \tr_{A_I A_O B_I B_O} \left(W^{T_{A_I A_O B_I B_O}} (M_x^{A_I A_I' A_O A_O'} \otimes M_y^{B_I B_I' B_O B_O'}) \right)
\end{equation}
is the Choi state of a CPTP map from $P A_I' B_I'$ to $F A_O' B_O'$, i.e. $\tr_{F A_O' B_O'} G_{xy} = \id^{PA_I' B_I'}$. In the above, $W^{T_{A_I A_O B_I B_O}}$ is the partial transpose of $W$ on the $A_I, A_O, B_I, B_O$ Hilbert spaces, while $M_x$ and $M_y$ are the Choi operators corresponding to the CPTP maps $\mathcal{M}_x^{A_I A_I' A_O A_O'}$ and $\mathcal{M}_y^{B_I B_I' B_O B_O'}$.
\end{definition}

This view of processes as  a \textit{supermaps} $M_x \otimes M_y \mapsto G_{xy}$ allows one to define pure processes~\cite{Araujo2017}, of which we recall the definition. 
\begin{definition} (Pure process)
\label{def:pure_process}
A process matrix $W^{PFA_I A_O B_I B_O}$ is pure if, for all ancillary Hilbert spaces $A_I', A_O', B_I', B_O'$~\footnote{The dimensions of the primed Hilbert spaces must satisfy $d_{A_I} d_{A_I'} = d_{A_O} d_{A_O'}$, $d_{B_I} d_{B_I'} = d_{B_O} d_{B_O'}$, and $d_P d_{A_I'} d_{B_I'} = d_F d_{A_O'} d_{B_O'}$}, and all unitaries $U : A_I A_I' \to A_O A_O'$, $V: B_I B_I' \to B_O B_O'$, the resulting transformation
\begin{equation}
\label{eq:G_UV}
G_{UV} = \tr_{A_I A_O B_I B_O} \left(W^{T_{A_I A_O B_I B_O}} |U \rrangle \llangle U| \otimes |V\rrangle \llangle V| \right)
\end{equation}
is the Choi state of a unitary channel from $P A_I' B_I'$ to $F A_O' B_O'$.
\end{definition}

Purifiable processes are processes that can be obtained from some pure process after tracing out certain degrees of freedom. In contrast to the familiar situations in quantum information, where through the use of an ancillary Hilbert space, any mixed state can be purified and any quantum channel can be dilated to a unitary channel, there exists processes that cannot be purified~\cite{Araujo2017}. Purifiable processes have been argued to be more reasonable physically, because the irreversiblity that occurs within them can be interpreted as arising from forgetting degrees of freedom in a fundamentally reversible process. In this section we obtain another justification for the reasonableness of pure processes: those are precisely the processes that admit a description in terms of causal frames of reference.

We collect here an important characterisation of pure processes, whose proof is provided in Ref.~\cite{Araujo2017}.
\begin{theorem}
\label{thm:pure_rank1}
A process $W$ is pure if and only if $W = |U_w \rrangle \llangle U_w|$ for some unitary $U_w : P A_O  B_O \to F  A_I B_I$.
\end{theorem}
We stress that the above theorem does not mean that all unitaries $U : P A_O B_O \to F A_I  B_I$ are such that $|U \rrangle \llangle U |$ is a process.

Theorem~\ref{thm:pure_rank1} allows to simplify the expression for $G_{UV}$ in Eq.~(\ref{eq:G_UV}). Let $W = |w \rangle \langle w|$ be a pure process, and define
\begin{equation}
\label{eq:G_UV_ket}
|\mathcal{G}(U,V)\rrangle^{P A_I' B_I', F A_O' B_O'} := |w \rangle^{T_{A_I A_O B_I B_O}} \cdot |U \rrangle^{A_I A_I' A_O A_O'} |V \rrangle^{B_I B_I' B_O B_O'},
\end{equation}
where $|w \rangle^{T_{A_I A_O B_I B_O}} : A_I A_O B_I B_O \to P F$ is the matrix obtained by partial transpose of $|w \rangle$. Then we have that
\begin{equation}
G_{UV} = |\mathcal{G}(U,V) \rrangle \llangle \mathcal{G}(U,V)|.
\end{equation}

We make a few comment about the dimensions of the Hilbert space. We first observe that no loss of generality occurs by restricting our attention to pure processes in which $d_{A_I} = d_{A_O}$, $d_{B_I} = d_{B_O}$ and $d_P = d_F$. Indeed, suppose $W^{P F A_I A_O B_I B_O}$ is pure a process for which the input-output dimensions do not match. We can just add new Hilbert spaces $A_I', A_O', B_I', B_O'$ to make the dimension match. We define
\begin{equation}
\tilde{W} = W^{PF A_I A_O B_I B_O} \otimes |\id \rrangle \llangle \id|^{P_A A_I'} \otimes |\id \rrangle \llangle \id|^{P_B B_I'} \otimes |\id \rrangle \llangle \id|^{A_O' F_A} \otimes |\id \rrangle \llangle \id|^{B_O' F_B},
\end{equation}
where $P_A \cong A_I'$ , $P_B \cong B_I'$ , $F_A \cong A_O'$ , $F_B \cong B_O'$. The new process $\tilde{W}$ is pure and acts on the Hilbert spaces $\tilde{P} = P P_A  P_B$ , $\tilde{F} = F F_A  F_B$ , $\tilde{A}_I = A_I  A_I'$ , $\tilde{A}_O = A_O  A_O'$ , $\tilde{B}_I = B_I  B_I'$ , $\tilde{B}_O = B_O B_O'$, where now the input and output Hilbert spaces have the same dimension. We can recover $W$ from $\tilde{W}$ by tracing out over the primed Hilbert spaces. A second observations is that Theorem~\ref{thm:cheung} implies that the dimensions $d_{A}:= d_{A_I} = d_{A_O}$ and $d_{B}:=d_{B_I} = d_{B_O}$ must be divisors of $d_P = d_F$ in order for a process to be pure.

\begin{definition}(The induced map of a pure process)
\label{def:induced_map}
Let $W$ be a pure process with $d_{A_I} = d_{A_O} $ , $d_{B_I} = d_{B_O}$, $d_P = d_F$. The induced map $\mathcal{G}$ is the bilinear map that sends pairs of unitaries $U: A_I \to A_O$ to $V: B_I \to B_O$ to a unitary $\mathcal{G}(U,V): P \to F$, defined by
\begin{equation}
|\mathcal{G}(U,V)\rrangle^{PF} := |w \rangle^{T_{A_I A_O B_I B_O}} \cdot |U \rrangle^{A_I A_O} |V \rrangle^{B_I B_O},
\end{equation}
\end{definition}

Processes where parties have the same input and output Hilbert space dimension are fully determined by their action on unitaries:
\begin{proposition}
\label{prop:G_enough}
Let $W= |w \rangle \langle w|$ be a pure bipartite process with $d_{A_I} = d_{A_O}=: d_A$, $d_{B_I} = d_{B_O}=: d_B$, $d_P = d_F$, and let $\mathcal{G}$ be it's induced map as in Def.~\ref{def:induced_map}. Let $\{U_i \}_{i=1}^{d_{A}^2}$, $\{V_j \}_{j=1}^{d_{B}^2}$ be orthonormal bases of unitaries (see Appendix~\ref{app:CJ}) for Alice's and Bob's Hilbert space, respectively. Then
\begin{equation}
|w \rangle = \frac{1}{d_A d_B}\sum_{i,j} |\mathcal{G}(U_i,V_j)\rrangle^{PF} |U_i^* \rrangle^{A_I A_O} |V_j^* \rrangle^{B_I B_O}.
\end{equation}
\end{proposition}

\begin{proof}
Since $|w \rangle$ is a pure state, and that $|U_i^* \rrangle  \otimes |V_j^*\rrangle$ forms a basis for $A_I \otimes A_O \otimes B_I \otimes B_O$, $|w\rangle$ can be expanded as
\begin{equation}
|w\rangle = \sum_{i,j} |\psi_{i,j}\rangle^{PF} |U_i^*\rrangle^{A_I A_O} |V_j^* \rrangle^{B_I B_O},
\end{equation}
where $|\psi_{i,j} \rangle$ are vectors that we must determine. Eq.~(\ref{eq:G_UV_ket}) together with Eq.~(\ref{eq:transp_pure_choi}) yields
\begin{equation}
|\mathcal{G}(U_i,V_j)\rrangle^{PF}  =  |w\rangle^{T_{A_I A_O B_I B_O}} |U_i \rrangle^{A_I A_O}  |V_j \rrangle^{B_I B_O} =d_A d_B |\psi_{i,j} \rangle^{PF}.
\end{equation} \qed
\end{proof}

\begin{theorem} 
\label{thm:frames_implies_pure}
Every pair of compatible causal reference frame as in Def.~\ref{def:process_rel}, $\mathcal{G}(U_A,U_B) = \Phi_A(U_B) (U_A \otimes \id^{E_A} ) \Pi_A(U_B) = \Phi_B(U_A) (U_B \otimes \id^{E_B} ) \Pi_B(U_A)$, defines a valid pure process
\begin{equation}
|w \rangle = \frac{1}{d_A d_B} \sum_{i,j} |\mathcal{G}(U_i,V_j)\rrangle^{PF} |U_i^*\rrangle^{A_I A_O} |V_j^*\rrangle^{B_I B_O},
\end{equation}
where $\{U_i : A_I \to A_O\}_{i = 1}^{d_A^2}$ and $\{V_j: B_I \to B_O \}_{j=1}^{d_B^2}$ are a basis of orthonormal unitaries.
\end{theorem}
\begin{proof}
Let $|w\rangle$ be as above, let $A_I' \cong A_O'$, $B_I' \cong B_O'$ be ancillary Hilbert spaces of any dimension, and let $M: A_I A_I' \to A_O A_O'$ and $N: B_I B_I' \to B_O B_O'$ be unitaries. These can be expanded in a basis as
\begin{align}
\label{eq:M_basis}
M = \sum_i  U_i \otimes a_i\\
\label{eq:N_basis}
N = \sum_i V_j \otimes b_j
\end{align}
where $\{U_i : A_I \to A_O\}_{i = 1}^{d_A^2}$ and $\{V_j: B_I \to B_O \}_{j=1}^{d_B^2}$ are a basis of orthonormal unitaries, and $a_i: A_I' \to A_O', b_j: B_I' \to B_O'$ are linear maps, not necessarily unitary. The transformation induced by $M, N$ is
\begin{align}
|\mathcal{K}(M,N)\rrangle^{P A_I' B_I' F A_O' B_O'} := & \, |w\rangle^{T_{A_I A_O B_I B_O}} |M \rrangle^{A_I A_I' A_O A_O'} |N \rrangle^{B_I B_I' B_O B_O'}\\
= & \sum_{i,j} |\mathcal{G}(U_i,V_j) \rrangle^{PF} |a_i\rrangle^{A_I' A_O'} |b_j\rrangle^{B_I' B_O'}
\end{align}
and showing that the process is valid and pure is equivalent to showing that
\begin{equation}
\label{eq:K_MN}
\mathcal{K}(M,N) = \sum_{i,j} \mathcal{G}(U_i,V_j) \otimes a_i \otimes b_j
\end{equation}
is a unitary from $P A_I'  B_I'$ to $F A_O' B_O'$.

We first write $\mathcal{G}(U_i,V_j)$ in Alice's causal reference frame
\begin{align}
\mathcal{K}(M,N) &= \sum_{i,j} \left( \Phi_A(V_j) (U_i \otimes \id^{E_A}) \Pi_A(V_j)\right) \otimes a_i \otimes b_j\\
&= \sum_j f_M(V_j) \otimes b_j,
\end{align}
where $P = A \otimes E_A$, and in the last line we defined $f_M(V) : P A_I' \to F A_O'$ as 
\begin{equation}
f_M(V) =\sum_i  \bigg(\Phi_A(V) (U_i \otimes \id^{E_A}) \Pi_A(V) \bigg) \otimes a_i =\left(\Phi_A(V) \otimes \id^{A_O'}\right)  \left(\sum_i U_i \otimes \id^{E_A} \otimes a_i\right)\left( \Pi_A(V)\otimes \id^{A_I'} \right).
\end{equation}
The second equality above, together with Eq.~(\ref{eq:M_basis}) makes it clear that $f_M(V)$ is unitary whenever $V$ is unitary, because it is a product of three unitaries. Notice also that $f_M(V)$ is a linear function both in $M$ and in $V$, so it is continuous in those two variables. Linearity in $V$ is proven by switching to Bob's causal frame:
\begin{equation}
\label{eq:f_M_Bob}
f_M(V) = \sum_i  \left(\Phi_B(U_i) (V \otimes \id^{E_B}) \Pi_B(U_i) \right)\otimes a_i.
\end{equation}

Therefore $f_M(V)$ satisfies the conditions of Theorem~\ref{thm:cheung}, and we conclude that 
\begin{equation}
\label{eq:f_M_new}
f_M(V) = S_M ((V \otimes \id_r) \oplus  (V^T \otimes \id_s)T_M,
\end{equation}
where $r + s =  d_{A'}d_P$ and where $T_M: P  A_I' \to A_I A_I'  E_A$ and $S_M: A_O A_I' E_A \to F  A_O'$ are some unitaries that depend on $M$. We wish to show that Eq.~\eqref{eq:f_M_new} holds without any transposes. To do so, we switch to the CJ representation and write
\begin{equation}
|f_M(V) \rrangle = Q_M |V \rrangle,
\end{equation}
where $Q_M: B_I B_O \to P A_I' F A_O'$ is a linear map whose properties are studied in Appendix~\ref{app:marcus_choi}. There it is shown that $\rho_M := \tr_{P A_I' B_I} |Q_M \rrangle \llangle Q_M|$ has $d_B^2$ times the eigenvalue $r$ and $d_B^2 s$ times the eigenvalue $1$.

Furthermore, note that the map that sends $M$ to the eigenvalues of $\rho_M$ is continuous (the map from $M$ to $f_M$ is continuous because it is linear, and the definition of $\rho_M$ from $f_M$ in Appendix~\ref{app:marcus_choi} only uses continuous functions). Since the set of allowed eigenvalues for $\rho_M$ is discrete (see Appendix~\ref{app:marcus_choi}), if two matrices $M_1$ and $M_2$ are connected by a continuous path in the space of matrices, then the eigenvalues of $\rho_{M_1}$ and $\rho_{M_2}$ will be equal. Since every unitary $M$ can be reached by a continuous path in the space of unitaries starting at the identity, $\rho_\mathcal{I}$ has the same eigenvalues as $\rho_M$. Finally, because the dimensions $r$ and $s$ are directly related to the eigenvalues of $\rho_M$, this implies that $r$ and $s$ in Eq.~\ref{eq:f_M_new} are constant for all $M$.

Therefore it suffices to verify that there are no transposes in Eq.~\eqref{eq:f_M_new} in the case where $M$ is the identity map $\mathcal{I}^{A_I \to A_O} \otimes \mathcal{I}^{A_I' \to A_O'}: |i \rangle^{A_I} |j \rangle^{A_I'} \mapsto |i \rangle^{A_O} |j \rangle^{A_O'}$. We get from Eq.~(\ref{eq:f_M_Bob}) that
\begin{equation}
f_{\mathcal{I}} (V) = \bigg(\Phi_B(\mathcal{I}) (V \otimes \id^{E_B}) \Pi_B(\mathcal{I})\bigg) \otimes \mathcal{I}^{A_I' \to A_O'},
\end{equation}
which shows that $f_{\mathcal{I}}$ takes the form of Eq.~(\ref{eq:f_M_new}), without transpose, i.e. $r = d_{E_A}$ and $s= 0$. Thus we have
\begin{align}
\mathcal{K}(M, N) &= \sum_j \left( S_M (V_j \otimes \id^{E_B A_I'}) T_M \right) \otimes b_j \\
& = (S_M\otimes \id^{B_O'}) \left(\sum_j V_j \otimes \id^{E_B} \otimes b_j \right) (T_M \otimes \id^{B_I'}),
\end{align}
and Eq.~(\ref{eq:N_basis}) shows that $\mathcal{K}(M,N)$ is unitary, since it is the product of three unitaries.
\qed
\end{proof}

The proof above can be generalised recursively to any number of parties, as we sketch here. Assume that every $N-1$ partite compatible causal reference frames defines a valid pure process and let $\mathcal{G}(U_{1}, ... , U_{N})$ be compatible reference frames for $N$ parties. Then we have to show that $\mathcal{K}(M_1, ... , M_N)$ defined as the obvious generalisation of Eq.~(\ref{eq:K_MN}) is unitary, where as before, $M_i = \sum_k U_i^k \otimes a_i^k$, for $i \in \{1, ... N\}$. Defining
\begin{equation}
\label{eq:J_M}
\mathcal{J}_{M_1,... ,M_{N-1}}(U_N) = \sum_{i_1, ..., i_{N-1}} \mathcal{G}(U_1^{i_1}, ... , U_{N-1}^{i_{N-1}}, U_{N}) \otimes a_1^{i_1} \otimes ... \otimes a_{N-1}^{i_{N-1}},
\end{equation}
we get
\begin{equation}
\label{eq:K_N_parties}
\mathcal{K}(M_1, ... , M_N) = \sum_k \mathcal{J}_{M_1,...,M_{N-1}}(U_{N}^k) \otimes a_{N}^k.
\end{equation}
Now by the recursion assumption, $\mathcal{J}_{M_1,... ,M_{N-1}}(U_N)$ is unitary, and it is linear in $U_N$ as can be seen by using $A_N$'s causal reference frame decomposition for $\mathcal{G}$ in the Eq.~(\ref{eq:J_M}). Therefore (using as before Theorem~\cite{Cheung2003} and a continuity argument to get rid of the potential transposes), there exists unitaries $S_{M_1,...,M_{N-1}},T_{M_1,...,M_{N-1}}$ which depend on $M_1, .... M_{N-1}$, and such that
\begin{equation}
\mathcal{J}_{M_1,...,M_{N-1}}(U_N) = S_{M_1,...,M_{N-1}}(U_N \otimes \id )T_{M_1,...,M_{N-1}}.
\end{equation}
Plugging this into Eq.~(\ref{eq:K_N_parties}) shows that $\mathcal{K}(M_1, ... , M_N)$ is unitary, which completes the proof by recursion.

We now provide an expression for pure processes that makes manifest the existence of the causal reference frame decomposition for one of the parties.

\begin{theorem}
\label{thm:w_A_frame}
The process vector $|w\rangle$ corresponding to a pair of consistent causal reference frames as in Def.~\ref{def:process_rel}, $\mathcal{G}(U_A,U_B) = \Phi_A(U_B) (U_A \otimes \id^{E_A} ) \Pi_A(U_B) = \Phi_B(U_A) (U_B \otimes \id^{E_B} ) \Pi_B(U_A)$ can be written such that the causal frame of one party (in the following, Alice's) appears explicitly as
\begin{equation}
|w \rangle = \frac{1}{d_B} \llangle \id|^{E_I E_O} \sum_j |\Pi_A(V_j) \rrangle^{P, A_I E_I} |\Phi_A(V_j)\rrangle^{A_O E_O, F} |V_j^* \rrangle^{B_I B_O},
\end{equation}
where $E_I, E_O$ are Hilbert spaces isomorphic to $E_A$, and where $\{V_j: B_I \to B_O \}_{j=1}^{d_B^2}$ is a basis of orthonormal unitaries.
\end{theorem}
\begin{proof}
From Thm.~\ref{thm:frames_implies_pure}, we may write
\begin{align}
|w \rangle &= \frac{1}{d_A d_B} \sum_{i,j} |\mathcal{G}(U_i,V_j)\rrangle^{PF} |U_i^*\rrangle^{A_I A_O} |V_j^*\rrangle^{B_I B_O} \\
&= \frac{1}{d_A d_B} \sum_{i,j} |\Phi_A(V_j) (U_i \otimes \id ) \Pi_A(V_j) \rrangle^{PF} |U_i^*\rrangle^{A_I A_O} |V_j^*\rrangle^{B_I B_O}.
\end{align}
We prove the statement by "expanding" $|w \rangle$ in the $\{|U_i^*\rrangle^{A_I A_O} \}$ basis:
\begin{align}
\llangle U_i^*|^{A_I A_O} |w \rangle &=\frac{1}{d_B} \sum_j |\Phi_A(V_j) (U_i \otimes \id ) \Pi_A(V_j) \rrangle^{PF} |V_j^* \rrangle^{B_I B_O} \\
&= \left( \llangle U_i^*|^{A_I A_O} \otimes \llangle \id|^{E_I E_O} \right)\left(\frac{1}{d_B} \sum_j |\Phi_A(V_j) \rrangle^{A_O E_O, F} |\Pi_A(V_j) \rrangle^{P, A_I E_I} |V_j^*\rrangle^{B_I B_O} \right),
\end{align}
where in the second line we used Prop.~\ref{prop:mult_ket}, and where $E_I, E_O$ are two isomorphic copies of $E_A$. \qed
\end{proof}
This theorem can also be straightforwardly generalised to any number of parties.

Finally, we prove the converse of Thm.~\ref{thm:frames_implies_pure}, thus showing that pure processes are equivalent to causal frames of reference.
\begin{theorem}
\label{thm:pure_implies_frames}
If $W$ is a pure process with matching input and output dimensions $d_{A_I} = d_{A_O}$, $d_{B_I} = d_{B_O}$, $d_P = d_F$, then its induced map $\mathcal{G}$ admits a decomposition into causal frames as in Def.~\ref{def:process_rel}.
\end{theorem}

\begin{proof}
Have all parties except one of them (here we take Alice w.l.o.g.) perform a fixed unitary; in the bipartite case there is only Bob performing the fixed unitary $U_B$, but the argument applies for any number of parties. Then $\mathcal{G}(\cdot, U_B)$ defines a linear function that maps the unitaries of Alice to unitaries from $P$ to $F$. Theorem~\ref{thm:cheung} applied to the map $U_A \mapsto \mathcal{G}(U_A, U_B)$ gives us that
\begin{equation}
\label{eq:G_new}
\mathcal{G}(U_A, U_B) = \Phi_A(U_B) \left((U_A \otimes \id) \oplus (U_A^T \otimes \id) \right) \Pi_A(U_B),
\end{equation}
where the Hilbert space $P$ decomposes as $P = (A^1_I \otimes C_I) \oplus (A^2_I \otimes D_I)$ where $A^1_I, A^2_I$ are copies of $A$ and $C,D$ are Hilbert spaces whose dimensions satisfy $(d_C + d_D) d_A = d_P$. Similarly $F = (A^1_O \otimes C_O) \oplus (A^2_O \otimes D_O).$

We show by way of contradiction that $d_D$ must be zero because otherwise the process would not send arbitrary CPTP maps to CPTP maps. Define the one-party process
\begin{equation}
\mathcal{G}'(U_A) = \Phi_A^\dagger(U_B) \mathcal{G}(U_A, U_B) \Pi_A(U_B)^\dagger = \left((U_A \otimes \id) \oplus (U_A^T \otimes \id) \right). 
\end{equation}
We have to show that this process maps CPTP maps to CPTP maps. Suppose we apply this process to the map $M^{A_I A_O} = \id^{A_I} \otimes |0 \rangle \langle 0|^{A_O}$. Let us define $K_i  : = |0\rangle^{A_O} \langle i|^{A_I}$ and notice that $M = \sum_i |K_i \rrangle \llangle K_i|$. We also have that $\mathcal{G'}(K_i) = |0 \rangle \langle i| \otimes \id_r \oplus |i \rangle \langle 0| \otimes \id_s$ and
\begin{equation}
|\mathcal{G'}(K_i) \rrangle = |0 \rangle^{A_O^1} |i \rangle^{A_I^1} | \id \rrangle^{C_I C_O}\oplus |i \rangle^{A_O^2} |0 \rangle^{A_I^2} \otimes |\id\rrangle^{D_I D_O}.
\end{equation}
Therefore the resulting map after contracting the process with M will be
\begin{equation}
N := \tr_{A_I A_O} \left(W^{T_{A_I A_O}} M^{A_I A_O} \right)= \sum_i |\mathcal{G}'(K_i) \rrangle \llangle \mathcal{G}'(K_i)|^{PF}.
\end{equation}
For the process to be valid it must be the case that $N$ is CPTP, i.e. $\tr_F = \id^P$. One may check that
\begin{equation}
\tr_F N = \id^{A_I^1} \id^{C_O} \oplus \, d_A |0 \rangle \langle 0|^{A_I^2} \otimes \id ^{D_O},
\end{equation}
which is CPTP if and only if the $D$ Hilbert space has zero-dimension i.e. the causal reference frame description
\begin{equation}
\mathcal{G}(U_A, U_B) = \Phi_A(U_B)( U_A \otimes \id ) \Pi(U_B)
\end{equation}
holds for Alice. Repeating for all other parties completes the proof.
\qed

\end{proof}

\section{The causal reference frames of causal inequality violating processes}
\label{sec:swiss}

In this section we investigate the causal reference frames description of some processes that can violate causal inequalities. An interesting pure tripartite process which is known to violate causal inequalities was already studied in Refs.~\cite{Baumeler2015, Araujo2017, Araujo2017_CTC}. Written as a process vector, it is equal to
\begin{equation}
|w \rangle = \sum_{\mathbf{x,y}}|\mathbf{y}\rangle^P |\mathbf{x}\rangle^{O} |\mathbf{y} \oplus f(\mathbf{x}) \rangle^{I} |\mathbf{x}\rangle^{F},
\end{equation}
where $I = A_I B_I C_I$, $O = A_O B_O C_O$, where we use bold-face notation for three-component binary vectors and where
\begin{equation}
f(a,b,c) = (0,0,0) + (1,0,0) \delta_{b,0} \delta_{c,1} + (0,1,0) \delta_{a,1} \delta_{c,0} + (0,0,1) \delta_{a,0} \delta_{b,1}.
\end{equation}
Alternatively, we can describe $|w\rangle$ via it's induced map $\mathcal{G}(U_A, U_B, U_C)$
\begin{align}
\label{eq:1st_G_swiss}
\mathcal{G}(U_A,U_B,U_C) (U_A^\dagger \otimes U_B^\dagger \otimes U_C^\dagger) |i i i \rangle &= |i i i \rangle \\
\mathcal{G}(U_A,U_B,U_C)(XU_A^\dagger \otimes U_B^\dagger \otimes U_C^\dagger)|i0 1 \rangle &= |i0 1 \rangle\\
\mathcal{G}(U_A,U_B,U_C)(U_A^\dagger \otimes XU_B^\dagger \otimes U_C^\dagger)|1 i 0 \rangle &= |1 i 0 \rangle\\
\label{eq:last_G_swiss}
\mathcal{G}(U_A,U_B,U_C) (U_A^\dagger \otimes U_B^\dagger \otimes XU_C^\dagger)|0 1 i \rangle &= |0 1 i \rangle,
\end{align}
where $i \in \{0,1\}$. When described from Alice's event-frame, it is
\begin{equation}
\Qcircuit @C=.5em @R=0em @!R {
& \qw & \targ & \gate{U_A} & \ctrl{1} & \ctrlo{1} & \qw \\
& \gate{U_B}& \ctrlo{-1} & \qw & \gate{U_B X U_B^\dagger} &\ctrl{1} & \qw \\
& \gate{U_C} & \ctrl{-1} & \qw&  \ctrlo{-1} & \gate{U_C X U_C^\dagger}& \qw
}
\label{eq:circuit_swiss}
\end{equation}
Here $X$ is the Pauli-X operator, and a white circle is a control by the $|0\rangle$ state. This process has the curious feature that the past $\Pi_A$ is linear in $U_B$ and $U_C$, but the future $\Phi_A$ still depends non-trivially on $U_B, U_C$. Interestingly, this process violates causal inequalities even under the restriction to classical instruments (diagonal in the computational basis)~\cite{Baumeler2015}.

We can obtain another valid process by taking the time reverse of $|w \rangle$, as explained in Appendix~\ref{app:time_rev}. The result is
\begin{align}
\label{eq:rev_swiss_ket}
|w _r \rangle& = \sum_{\mathbf{x,y}}|\mathbf{x}\rangle^P|\mathbf{y} \oplus f(\mathbf{x}) \rangle^{O}  |\mathbf{x}\rangle^{I} |\mathbf{y}\rangle^{F} \\
&= \sum_{\mathbf{x,y}}|\mathbf{x}\rangle^P|\mathbf{y} \rangle^{O}  |\mathbf{x}\rangle^{I} |\mathbf{y}  \oplus f(\mathbf{x})\rangle^{F}\\
&= \sum_{\mathbf{x,y}} |\mathbf{y}\rangle^P \ket{\mathbf{x}}^O \ket{\mathbf{y}}^I \ket{\mathbf{x} \oplus f(\mathbf{y})}^F,
\end{align}
where in the second line we made the change $\mathbf{y} \mapsto \mathbf{y} \oplus f(\mathbf{x})$ and in the third line we relabelled $\mathbf{x} \leftrightarrow \mathbf{y}$. Equivalently, this process can be described with its induced map as
\begin{align}
\mathcal{G}_r(U_A,U_B,U_C)|i i i \rangle &= (U_A \otimes U_B \otimes U_C) |i i i \rangle \\
\mathcal{G}_r(U_A,U_B,U_C)|i0 1 \rangle &= (X U_A \otimes U_B \otimes U_C) |i0 1 \rangle\\
\mathcal{G}_r(U_A,U_B,U_C)|1 i 0 \rangle &= (U_A \otimes  X U_B \otimes U_C)|1 i 0 \rangle\\
\mathcal{G}_r(U_A,U_B,U_C)|0 1 i \rangle &= (U_A \otimes U_B \otimes X U_C) |0 1 i \rangle,
\end{align}
where $i \in \{0,1\}$. At first sight it seems that the transformation $\mathcal{G}_r$ can be understood causally: the parties parallely apply $U_A, U_B, U_C$ on the input quantum state $|\psi\rangle^P$, and then a Pauli-X gate is applied to the state in a way that depends on the state in the past $|\psi \rangle$. Indeed, in classical theory, this process has a simple realisation: first copy the input state, then parallely apply the transformations $U_A \otimes U_B \otimes U_C$ on the original state, and finally apply a controlled gate from the copy to the target. Of course, this particular strategy is forbidden in quantum mechanics because of the no-cloning theorem. 

The causal reference frames description of $|w_r\rangle$ however tells a different story. When written in Alice's causal reference frame, the process is
\begin{equation}
\Qcircuit @C=.5em @R=0em @!R {
& \ctrlo{1} & \ctrl{1}  & \gate{U_A} &\targ & \qw & \qw\\
&\ctrl{1} &\gate{U_B^\dagger X U_B} & \qw & \ctrlo{-1}& \gate{U_B} & \qw\\
&\gate{U_C^\dagger X U_C} & \ctrlo{-1} & \qw & \ctrl{-1}& \gate{U_C} & \qw
}
\label{eq:circuit_rev_swiss}
\end{equation}
which has the same feature that was previously noticed for $|w\rangle$ (now it is the future $\Phi_{A}$ that is linear in $U_B, U_C$, while $\Pi_A$ has non-trivial dependence on $U_B$ and $U_C$).

For completeness we note that the process $|w_r\rangle$ can also be written as a circuit containing linear post-selected closed timelike curves (CTCs)~\cite{Araujo2017_CTC}
\begin{equation}
\includegraphics[width=6cm]{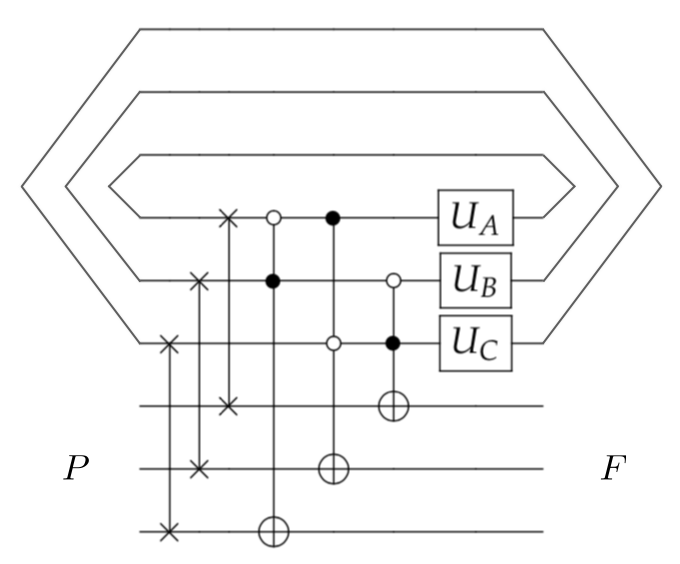}
\centering.
\label{eq:rev_swiss_CTC}
\end{equation}
In the above circuit, each loop can be (probabilistically) implemented by, on the left hand side of the loop, preparing a maximally entangled state $|\Phi^+ \rangle = \sum_i |i \rangle |i\rangle$, and on the right hand side, performing a Bell measurement and post-selecting on the outcome $|\Phi^+ \rangle$. We refer the reader to Ref.~\cite{Araujo2017_CTC} for a more complete discussion.

We now turn to the question of whether $|w_r\rangle$ can be used to violate causal inequalities. If the input state in the past is 
\begin{equation}
\ket{\psi}^P = \sum_\mathbf{u} \psi_\mathbf{u} \ket{\mathbf{u}},
\end{equation}
then we define the reduced tripartite process matrix $W_\psi \in \mathcal{L}(I \otimes O)$ by
\begin{equation}
W_\psi = \tr_{PF}\left( |\psi \rangle \langle \psi|^P \cdot |w_r \rangle  \langle w_r | \right)=  \sum_{\mathbf{u,v,x}} \psi_\mathbf{u}^* \psi_\mathbf{v} |\mathbf{u}\rangle \langle \mathbf{v}|^I\otimes |\mathbf{x}+f(\mathbf{u})\rangle \langle \mathbf{x} + f(\mathbf{v})|^O.
\end{equation}
A simple choice of input state is the uniform superposition $|\psi\rangle^P = \frac{1}{2\sqrt{2}}\sum_\mathbf{u} |\mathbf{u}\rangle$, which yields
\begin{equation}
W_\psi = \frac{1}{8}\sum_{\mathbf{u,v,x}}  |\mathbf{u}\rangle \langle \mathbf{v}|^I\otimes |\mathbf{x}+f(\mathbf{u})\rangle \langle \mathbf{x} + f(\mathbf{v})|^O
\end{equation}
This process violates the causal inequality
\begin{equation}
I_1 = P_{AB}(11|110) + P_{BC}(11|011) + P_{AC}(11|101) - P_{ABC}(111|111) \geq 0
\end{equation} 
described in Ref.~\cite{Abbott2016}. The strategy that achieves the violation was found by performing a seesaw optimisation~\cite{Branciard2016, Werner2001} on the parties' instruments:
\begin{align}
M_{0|0}^{A_I A_O}& =M_{0|0}^{B_I B_O}= M_{0|0}^{C_I C_O} =\frac{1}{2} \left( \id - \frac{1}{2} \id  \otimes X - \frac{1}{2} X \otimes X \right) \nonumber \\
 M_{1|0}^{A_I A_O}& =M_{1|0}^{B_I B_O}= M_{1|0}^{C_I C_O} = 0  \nonumber \\
M_{0|1}^{A_I A_O}& =M_{0|1}^{B_I B_O}= M_{0|1}^{C_I C_O} \approx \frac{1}{4}\left(\id - 0.97926 X  \otimes \id - 0.20258 Y  \otimes \id \right) \nonumber \\
M_{1|1}^{A_I A_O}& =M_{1|1}^{B_I B_O}= M_{1|1}^{C_I C_O} \approx \frac{1}{4} \bigg(\id + \id  \otimes X + 0.97926(X \otimes \id + X  \otimes X) + 0.20258 (Y  \otimes \id + Y  \otimes X) \bigg),
\end{align}
where $X,Y$ are Pauli matrices. We could not find a closed-form expression for the two numbers appearing above. The value of the violation that we obtain is $I_1 \approx -\frac{1}{4}$. The algebraic violation for this inequality is $I_1^{max} = -1$ and can be attained with process matrix correlations~\cite{Abbott2016}.

The reasons why $|w\rangle$ and $|w_r\rangle$ can violate causal inequalities seem to be fundamentally different. In the case of $|w\rangle$, the reason appears to come from a ``classical'' non-causal influence of the future on the past, since $|w\rangle$ still violates causal inequalities when seen as a classical process matrix~\cite{Baumeler2015}. However, in the case of $|w_r\rangle$, the phenomena seems related to the no-cloning theorem: quantum mechanics restricts the ways in which the future can depend on the past. The picture of the process $|w_r \rangle$ in terms of causal reference frames might help understanding the similarities between $|w_r\rangle$ and $|w\rangle$, and why they both lack a known physical realisation. Indeed, they both have the same feature: according to Alice's causal frame of reference, the future $\Phi_A(U_B, U_C)$ (resp. the past $\Phi_A(U_B, U_C)$) "contains" Bob and Charlie's events, in the sense that it is linear in both $U_B$ and $U_C$. However, despite this fact, the past $\Pi_A$ (resp. the future $\Phi_A$) still depends non-trivially on $U_B, U_C$.

In general, if $\Phi_A$ (resp. $\Pi_A$) is independent of $U_B, U_C$, then $\Pi_A$ (resp. $\Phi_A$) must be linear in $U_B, U_C$ in order for the process to be linear. Common sense intuitions about causally would imply that the converse is also true: that $\Phi_A$ being linear in $U_B, U_C$ should imply that $\Pi_A$ is independent of $U_B, U_C$. Indeed it seems reasonable to interpret, for example, the linearity of $\Phi_A$ in $U_B$ as meaning that Bob is localised in the future of $U_B$. Moreover, for the case of the quantum switch -- the only known example of a physically realisable non-causal process-- , something similar as the above holds. In the quantum switch as described from Alice's causal reference frame, Bob's event is delocalised in time, but one can reason as ``if the control qubit is in state $|0\rangle$, Bob is in the past, while if the control is $|1\rangle$ then Bob is in the future''. More formally, there exists projectors $|0\rangle \langle 0|^{E_A}$ and $|1\rangle \langle 1|^{E_A}$, such that $ \Phi_A(U_B) (\id^{A_I} \otimes |0 \rangle \langle 0|^{E_A})$ is linear in $U_B$, while $(\id^{A_O} \otimes |0 \rangle \langle 0|^{E_A}) \Pi_A(U_B)$ is independent of $U_B$ and such that $ (\id^{A_I} \otimes |1 \rangle \langle 1|^{E_A})\Pi_A(U_B)$ is linear in $U_B$, while $\Phi_A (U_B) (\id^{A_O} \otimes |1 \rangle \langle 1|^{E_A})$ is independent of $U_B$. 

We believe that the observations above could be formalised into a notion of "weakly causal processes", that would include causally ordered processes and the quantum switch as special cases, but not the processes $|w\rangle$ and $|w_r\rangle$. This would yield a more physical way -- beyond the violation of causal inequalities -- of explaining why $|w \rangle$ and $|w_r \rangle$ possess a stronger type of non-causality than the quantum switch.

\section{Comments on the link with time-delocalised subsystems}
\label{sec:oreshkov}

We comment on the link between our causal frames of reference and the time-delocalised subsystems introduced by Oreshkov in Ref.~\cite{Oreshkov2018}. We also prove that all multipartite pure processes admit a representation in terms of time-delocalised subsystems. Note however, that the existence of such a representation does not imply that all such processes can be physically realised (for example the processes of Eqs.~\eqref{eq:circuit_swiss} and \eqref{eq:circuit_rev_swiss} do not have a known physical realisation).

Let $|w\rangle$ be a $N$-partite pure process, with parties $A^1, A^2, ... A^N$ whose input and ouput Hilbert spaces have equal dimension: $A_I^{k} \cong A_O^{k}$. By making one of the parties (say $A^1$, w.l.o.g) perform a unitary $|U\rrangle^{A_I^1 A_O^1}$, we obtain the reduced $(N-1)$ party process
\begin{equation}
\label{eq:induced_w_oreshkov}
|\xi(U) \rrangle := \llangle U^*|^{A_I^1 A_O^1} |w\rangle.
\end{equation}
Since $|w\rangle$ is pure, $|\xi(U) \rrangle$ is also a pure process, and by Theorem~\ref{thm:pure_rank1} it is the Choi state of a unitary channel from $P \otimes A_O^2 \otimes ... \otimes A_O^N$ to $F \otimes A_I^2 \otimes ... \otimes A_I^N$. More generally, if $A^1$ performs a general CPTP map $M^{A_I^1 A_O^1}$, we have that $\tr_{A_I^1 A_O^1} \left(M^{A_I^1 A_O^1} |w \rangle \langle w|\right)$ is the Choi state of a CPTP map from $P A_O^2,..., A_O^N$ to $F, A_I^2,... A_I^N$~\footnote{This fact can be proved using the fact that $\tr_{A_I^1...A_I^N A_O^1... A_O^N F}\left(M^{A_I^1 A_O^1} \otimes \id^{A_I^1...A_I^N} \otimes \rho_2^{A_O^2} \otimes ... \otimes \rho_N^{A_O^N} | w \rangle \langle w | \right) = \id^P$ for all states $\rho_2^{A_O^2}, ... , \rho_N^{A_O^N}$.}. Therefore, the map $U \mapsto \xi(U)$ can also be interpreted as defining a pure single partite process, whose past Hilbert space is $\tilde{P} := P \otimes A_O^2 \otimes ... \otimes A_O^N$ and whose future Hilbert space is $\tilde{F} := F \otimes A_I^2 \otimes ... \otimes A_I^N$. The single-partite version of Theorem~\ref{thm:pure_implies_frames} allows us to find unitaries $T: \tilde{P} \to A_I \otimes E$ and $S: A_O \otimes E \to \tilde{F}$, where $E$ is a Hilbert space of dimension $d_E = d_A d_{\tilde{P}}$, such that
\begin{equation}
\xi(U) = S ( U \otimes \id^E) T.
\end{equation}
Equivalently,
\begin{equation}
\label{eq:w_time_deloc}
|w\rangle = \llangle \id|^{E E'} |T \rrangle^{\tilde{P}, A_I E} |S \rrangle^{A_O E', \tilde{F}},
\end{equation}
where $E'$ is an isomorphic copy of $E$. From the above equation it is manifest that $A_I$ is maximally entangled with some subspace $\tilde{A}_I$ of $\tilde{P}$ and that $A_O$ is maximally entangled with some subspace $\tilde{A}_O$ of $\tilde{F}$. These subspaces $\tilde{A}_I, \tilde{A}_O$ are called \textit{time-delocalised subsystems} in Ref.~\cite{Oreshkov2018}. Eq.~(\ref{eq:w_time_deloc}) above answers in the affirmative the question posed in the conclusion of Ref.~\cite{Oreshkov2018}, concerning the existence of a representation in terms of time-delocalised subsystems for multipartite pure processes. 

Oreshkov's decomposition is obtained by looking at the family of reduced process that one gets by fixing one parties' choice of unitary, as in Eq.~(\ref{eq:induced_w_oreshkov}). Our causal-frames description is complementary, in the sense that it is obtained by considering the family of reduced processes that one gets by fixing the choice of unitary for $N-1$ parties. This yields a decomposition of the process as
\begin{equation}
|w \rangle = \llangle \id|^{E_A E_A'} \sum_{i_2, ..., i_N} |\Pi_A(U_{i_2}, ..., U_{i_N}) \rrangle^{P, A_I E_A} |\Phi_A(U_{i_2}, ..., U_{i_N}) \rrangle^{A_O E_A', F} |U_{i_2}^* \rrangle^{A_I^2 A_O^2}  ... |U_{i_N}^*\rrangle^{A_I^N A_O^N},
\end{equation}
where the sum is over orthonormal bases of unitaries.

The two decompositions are related (here in the bipartite case) via
\begin{align}
\llangle U_B^*|^{B_I B_O} \llangle \id|^{E E'} |S \rrangle^{P B_O, A_I E} |T \rrangle^{A_O E', F B_I} = \llangle \id |^{E_A E_A'} |\Pi_A(U_B) \rrangle^{P, A_I E_A} |\Phi_A(U_B)\rrangle^{A_O E_A', F},
\end{align}
which can also be shown graphically as
\begin{equation}
\includegraphics[width=14cm]{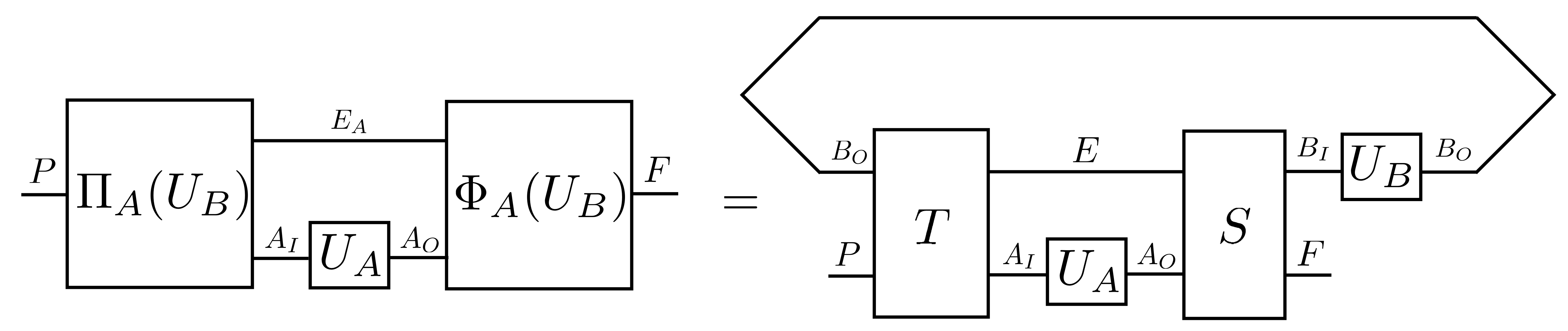},
\end{equation}
where the loop appearing in the circuit on the right hand side of the equation can be realised as a post-selected CTC, as discussed in the text after Eq.~\eqref{eq:rev_swiss_CTC} and in Ref.~\cite{Araujo2017_CTC}.

\section{The gravitational quantum switch}
\label{sec:gravity_switch}

The idea of causal reference frame that we introduced in this work can be used to analyse the gravitational quantum switch thought experiment~\cite{Zych2017}. We revisit this thought experiment, but making coordinates according to which Alice's event is localised, while Bob's event is delocalised and has a time coordinate that is entangled with the position of the mass. 

The Gedankenexperiment begins by considering the classical spacetime\footnote{A classical spacetime is an equivalence class, with respect to diffeomorphism, of tuples $(M, g, \Phi)$, where $M$ is a topological 4-manifold, $g$ is a semi-Riemannian metric and $\Phi$ describes the matter. The matter typically consists of a collection of fields, but $\Phi$ may also include the worldlines of point-like objects with negligible mass (test particles).} generated by a single massive spherically-symmetric object, and containing two localised laboratories -- Alice and Bob -- of negligible mass (such that they can be treated as test particles), whose worldlines are timelike curves $\lambda^A$, $\lambda^B$ which can be parametrised by the proper-time read by a clock inside the laboratories. Outside of the region occupied by the massive object, the metric is the Schwarzschild metric, which can be written in coordinates as
\begin{equation}
g= -\left(1 - \frac{2GM}{r c^2}\right) c^2 dt^2  + \left(1 - \frac{2GM}{r c^2}\right)^{-1}dr^2 + r^2 (d \theta^2 + \sin^2 \theta d\phi^2).
\end{equation}
With the above choice of coordinates, the position of the mass is fixed at $r=0$. We assume that the clocks are prepared such that they are initialised at $t=0$, and that their worldlines $\lambda_A , \lambda_B$ are held at the fixed coordinate angles $\theta = 0, \phi = 0$ and at fixed radial coordinates $R$ and $ R + h$, respectively. The proper time between two points $(t_1, r, 0 ,0)$ and $(t_2, r, 0,0)$ whose coordinate position is the same is $(t_2 - t_1) \sqrt{1 - \frac{2GM}{r c^2}}$, so that the worldlines of Alice and Bob are given, as a function of the proper time $\tau$ recorded by their respective clocks, by
\begin{align}
\lambda_1^A(\tau) = \left(t = \tau \left(1 - \frac{2GM}{R c^2} \right)^{-1/2}, r= R, \theta = 0, \phi = 0 \right)\\
\lambda_1^B(\tau) = \left(t = \tau \left(1 - \frac{2GM}{(R + h) c^2} \right)^{-1/2}, r= R + h, \theta = 0, \phi = 0 \right).
\end{align}
One can show~\cite{Zych2017} that there exists values of the parameters $R, h, M, \tau^*$ such that $\lambda_1^A(\tau^*)$ is in the causal future of $\lambda_1^B(\tau^*)$. We also consider a different spacetime, symmetrically related to the original one by a reflexion of the $z$-axis and a relabelling of the parties, in which the metric is identical but the worldlines of the parties are given instead by
\begin{align}
\lambda^A_2(\tau) = \left(t = \tau \left(1 - \frac{2GM}{(R + h)c^2} \right)^{-1/2}, r= R + h, \theta = \pi, \phi = 0 \right) \\
\lambda^B_2(\tau) = \left(t = \tau \left(1 - \frac{2GM}{R c^2} \right)^{-1/2}, r= R, \theta = \pi, \phi = 0 \right),
\end{align}
and we will then have that $\lambda_2^A(\tau^*)$ is in the causal past of $\lambda_2^B(\tau^*)$.

We will now make two different changes of coordinates (one for each of the two spacetimes), so that points on Alice's worldline corresponding to a particular value of the time read by Alice's clock have the same coordinate point in both spacetimes. Namely, we define $t_1 = t \left(1 - \frac{2GM}{R c^2} \right)^{1/2}$, $z_1 = r -R $ and $t_2 = t \left(1 - \frac{2GM}{(R+h) c^2} \right)^{1/2}$, $z_2 = r + R + h$. In the plane $\theta = \phi = 0$, we now have two different metrics which take the form
\begin{align}
\label{eq:schwarz1}
g_1 =- c^2 \frac{ \left(1 -  \frac{2GM}{(z + R) c^2}\right)}{\left(1 -  \frac{2GM}{R c^2}\right)} (dt)^2 + \left(1 - \frac{2GM}{(z + R) c^2}\right)^{-1}dz^2 \\
\label{eq:schwarz2}
g_2 =- c^2 \frac{ \left(1 -  \frac{2GM}{(z - R - h) c^2}\right)}{\left(1 -  \frac{2GM}{(R+h) c^2}\right)} (dt)^2 + \left(1 - \frac{2GM}{(z - R - h) c^2}\right)^{-1}dz^2,
\end{align}
where we use the same coordinate labels $z = z_1 =z_2$, $t = t_1 = t_2$ for the two different spacetimes. Note that at $z=0$, the $00$-component of the two metrics is equal to one, so that the coordinate time matches with Alice's proper time.~\footnote{This is a natural choice, but our only requirement on the coordinates is that Alice's events should occur at the same coordinate point in both spacetimes.}  We now have that $\lambda_1^B(\tau^*)$ occurs at a coordinate time smaller than $\tau^*$, while $\lambda_2^B(\tau^*)$ occurs at a coordinate time greater than $\tau^*$. Furthermore, there exists values of $R, h, M$ and $\tau^*$ (see Ref.~\cite{Zych2017} for details) such that $\lambda_1^B(\tau^*)$ is in the causal past of $\lambda_1^{A}(\tau^*)$ and $\lambda_2^B(\tau^*)$ is in the causal future of $\lambda_2^A(\tau^*)$. The situation in that case is depicted in Fig.~\ref{fig:gravity_switch}.

\begin{figure}[t]
\includegraphics[width=8cm]{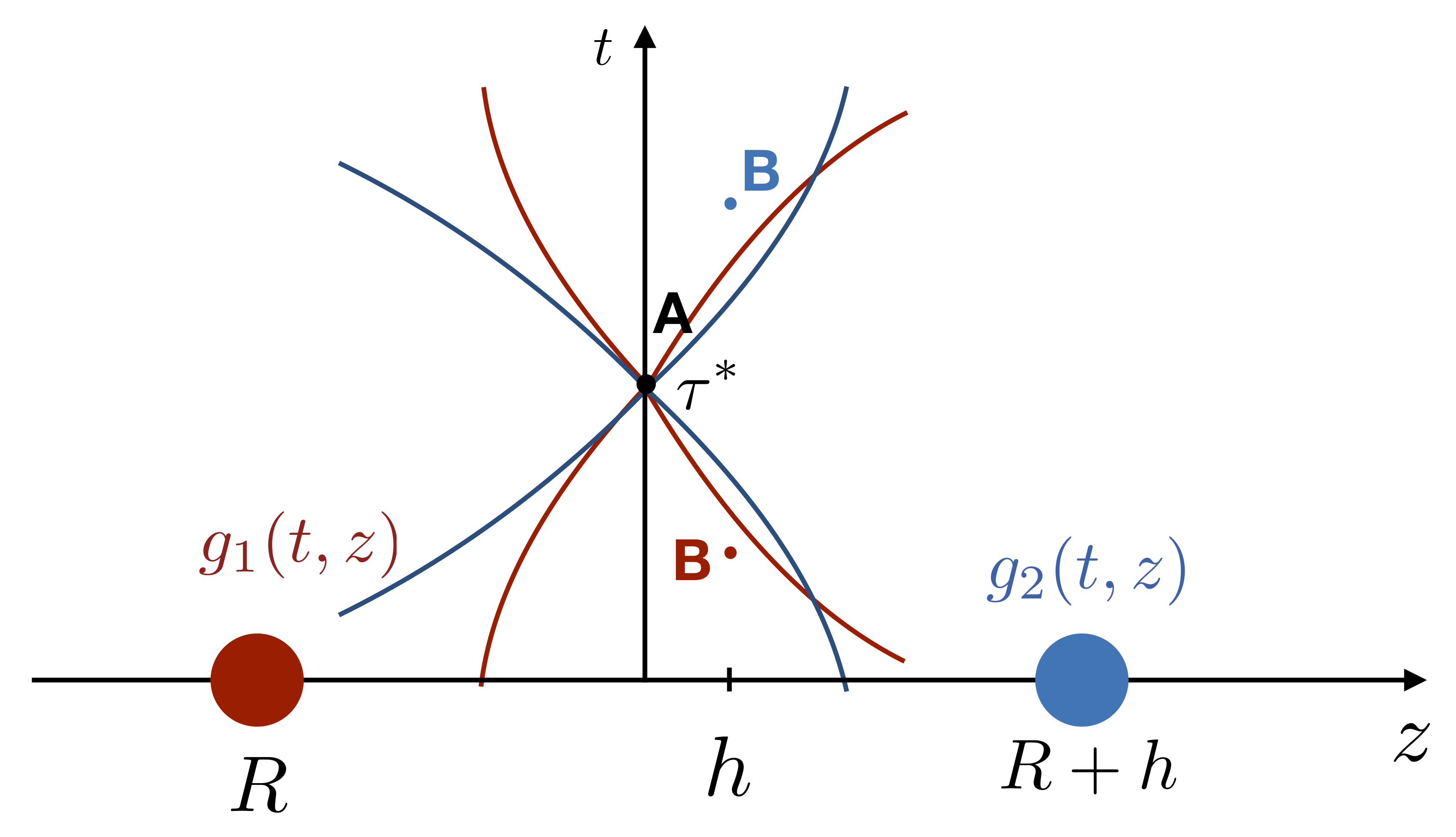}
\caption{Comparison of the two spacetimes $g_1, g_2$ using the coordinates described in the text. The relevant spacetime structures (the light-cone at $A$, the position of the mass, and the position of event $B$) for the metric $g_1$ are depicted in red, and in blue for the metric $g_2$. The light-cones at $A$ are curved due to the presence of matter. The events $A$ and $B$ are defined operationally with respect to local clocks. While event $A$ is represented at the same coordinate point for both spacetimes, the time at which event $B$ occurs depends on the position of the mass.}
\centering
\label{fig:gravity_switch}
\end{figure}

We now consider what happens when the initial state of the mass is in a superposition of the two positions (left and right). Under conservative physical assumptions, this should lead to a quantum superposition of the two spacetimes, which are described  in Alice's coordinates by the metrics $g_1, g_2$. These assumptions have been identified in Ref.~\cite{Zych2017} as 
\begin{enumerate}
\item Macroscopically distinct physical states can be assigned orthogonal quantum states.
\item For classical-like states, gravitational time dilation is correctly predicted by general relativity. 
\item The quantum superposition principle holds for all systems and at all scales. 
\end{enumerate}
Importantly, these assumptions are largely independent of the specific -- still unknown -- way in which gravity should be correctly quantised.

The representation of Fig.~\ref{fig:gravity_switch} is conceptually analogous to our previous treatment of the quantum switch in the causal frame of reference of event $A$. The position of the mass has the role of a quantum control for the coordinate point of event $B$. We assume that the state of the mass can be prepared by a party in the distant past, which we restrain for simplicity to a qubit subspace $\mathcal{H}_C$ spanned by the two classical-like states $|0 \rangle, |1 \rangle$, corresponding to the two positions for the mass in Fig.~\ref{fig:gravity_switch}. The parties at events $A$ and $B$ are receiving some localised quantum system with Hilbert space $\mathcal{H}_S$, on which they applying unitaries $U_A$, $U_B$. Then the evolution connecting states in the past to states in the future will be a unitary $\mathcal{G}(U_A, U_B)$ acting on the Hilbert space $\mathcal{H}_C \otimes \mathcal{H}_S$, and given by
\begin{equation}
\mathcal{G}(U_A, U_B) = |0 \rangle \langle 0|_C \otimes U_B U_A + |1 \rangle \langle 1|_C \otimes U_A U_B,
\end{equation}
which is precisely the quantum switch.

The only difference between the above treatment and the description given in Ref.~\cite{Zych2017} is that we use the coordinates $(t_1, z_1)$, and $(t_2, z_2)$ to compare the two classical spacetimes in Eqs.~(\ref{eq:schwarz1}, \ref{eq:schwarz2}), rather than using the Schwarzschild coordinates. More generally, we can perform an independent choice of coordinates for each of the two classical spacetimes that appear in the quantum superposition, and this yields different representations of the same physical situation. It would be interesting to investigate whether the usual diffeomorphism symmetry group of general relativity can be enlarged to take into account such "quantum-controlled" coordinate transformations.  A full theory of ``quantum coordinate changes'' does not exist, but steps have been taken in Ref.~\cite{Giacomini2017}, where quantum reference frames were studied in the context of Galilean relativity, and it was found that the spatial localisation of observers is frame-dependent.

\section{Conclusion}

We have introduced a local point of view on quantum causal structures, based on the relations between operationally defined events. To each event we have associated a causal reference frame in which that particular event is localised, and according to which it is possible to describe the process as the concatenation of unitaries according to some observer-dependent time. We have proven that there is an equivalence between pure processes and multilinear maps that admit a description in terms of consistent causal reference frames. We believe that this decomposition can be useful for further investigations of the pure process formalism, even if one does not endorse the physical interpretation proposed here.

We have studied the causal reference frames of non-causal processes, such as the quantum switch, and found that the time-localisation of events in such processes is in general observer-dependent. We have defined the time-reverse of a known tripartite causal inequality violating process~\cite{Baumeler2015, Araujo2017}, and we looked at the causal reference frames of both processes. We observed that these processes have a different structure than that of the quantum switch: for example one process has the property that in the causal reference frame of one event, the other events are arranged in such a way that they appear to be localised in the past (the past evolution depends linearly on the other parties' operations), while still having a non-trivial influence in the future. This suggests that the causal reference frames formalism can be a way to distinguish processes with a stronger form of non-causality (as witnessed by the violation of causal inequalities), from those whose non-causality is of a more benign nature such as the quantum switch and causally ordered processes. A more formal investigation of the various "types of non-causality", based on the language of causal reference frames, is deferred to future work.

As another application of causal reference frames, we have revisited a thought experiment in which the gravitational time dilation due to massive object in a quantum superposition of two position leads to a superposition of the causal order between events. We have shown that it is possible to describe the situation from the causal reference frame of one of the events by making appropriate coordinate transformations on the two classical spacetimes that appear in the superposition.

We finish with some comments about the possible relationships of our work with other approaches. There are superficial similarities between our framework and the framework of relative-locality~\cite{Amelino-Camelia2011}, in which a non-trivial geometry of momentum space leads to the observer-dependent locality of events. The ideas of relative locality have also been studied in the quantum regime~\cite{Amelino-Camelia2013}. It is currently unknown whether relative-locality allows for the violation of causal inequalities or the realisation of causally non-separable processes. Another interesting recent development is Hardy's operational reformulation of general relativity~\cite{Hardy2016}, and it is an open question whether our treatment of events in quantum causal structures can be reframed in his (potentially more general) formalism.

\section{Acknowledgements}

The authors acknowledge helpful discussions and comments from Mateus Ara\'{u}jo, \"{A}min Baumeler, Geoffroy Bergeron, Esteban Castro-Ruiz, Flaminia Giacomini and Philipp H\"{o}hn. We are grateful to David Trillo Fernandez for pointing out some mistakes in the proofs of the generalised Marcus theorem, and to Cyril Branciard for pointing out mistakes in Eqs.~\eqref{eq:1st_G_swiss}-\eqref{eq:last_G_swiss}. We acknowledge the support of the Austrian Science Fund (FWF) through the Doctoral Programme CoQuS and the project I-2526-N27 and the research platform TURIS, as well as support from the European Commission via Testing the Large-Scale Limit of Quantum Mechanics (TEQ) (No. 766900) project. P.A.G. acknowledges support from the Fonds de Recherche du Qu\'{e}bec -- Nature et Technologies (FRQNT). This publication was made possible through the support of a grant from the John Templeton Foundation. The opinions expressed in this publication are those of the authors and do not necessarily reflect the views of the John Templeton Foundation.


\providecommand{\href}[2]{#2}\begingroup\raggedright\endgroup

\appendix

\section{Choi-Jamio\l kowski isomorphism}
\label{app:CJ}

Let $A_I, A_O$ be Hilbert spaces of finite dimension $d_{A_I}$ and $d_{A_O}$, respectively, and let $\{i \rangle^{A_I} \}_{i=0}^{d_{A_I}-1}$, $\{ |j \rangle^{A_O}\}_{j=0}^{d_{A_O} - 1}$ be a choice of bases for $A_I$ and $A_O$. We denote by $\mathcal{L}(A_I), \mathcal{L}(A_O)$ the vector space of linear operators acting on $A_I, A_O$. 

We follow Refs.~\cite{Araujo2017,Araujo2017_CTC} in defining the Choi-Jamio\l kowski (CJ) isomorphism. We warn the reader that there exist different conventions in the literature. For any linear transformation $K : A_I \to A_O$, we define the ``double-ket''
\begin{equation}
|K \rrangle^{A_I A_O} = \sum_{i} |i \rangle^{A_I}\otimes (K |i \rangle)^{A_O} 
\end{equation}
Let $K, M : A_I \to A_O$. Then the inner product of $|K \rrangle$ and $|M \rrangle$ in $A_I A_O$ is equal to the Hilbert-Schmidt inner product of the operators $K, M$:
\begin{equation}
\llangle M | N \rrangle^{A_I A_O} = \tr \left( M^\dagger K \right).
\end{equation}
If $d_{A_I} = d_{A_O} =: d_A$, we will say that a set of unitaries $\{U_i : A_I \to A_O \}_{i = 1}^{d_A^2}$ is an orthonormal basis if $\{|U_{i}\rrangle \}_{i = 1}^{d_A^2}$is a basis for $A_I \otimes A_O$ and if $\llangle U_i | U_j \rrangle = d_A \delta_{ij}$.

We will often make use of the two easily verified identities
\begin{align}
\label{eq:choi_get_T}
|K \rrangle^{A_I A_O}&=\sum_i (K^T |i\rangle)^{A_I} \otimes |i \rangle^{A_O} \\ 
\label{eq:transp_pure_choi} |K \rrangle^T &= \llangle K^*|,
\end{align},
where $K^T: A_O \to A_I$ is the transpose of $K$, defined by $\langle i |^{A_I}K^T |j\rangle^{A_O}= \langle j |^{A_O} K |i \rangle^{A_I}$, and $K^*: A_I \to A_O$ is the complex conjugate $\langle j |^{A_O} K^{*} |i \rangle^{A_I} = \left(\langle j |^{A_O} K |i \rangle^{A_I}\right)^*$.

We also note that for any vector $|v\rangle^{A_I A_O}$, the isomorphism can be inverted to get the matrix $K_v$ for which $|K_v\rrangle = |v \rangle$. The explicit inversion formula is
\begin{equation}
K_v = |v \rangle^{T_{A_I}},
\end{equation},
where $T_{A_I}$ is the partial transpose on the $A_I$ Hilbert space, whose definition in the computational basis is
\begin{equation}
(|i\rangle^{A_I} |j\rangle^{A_O})^{T_{A_I}} = |j \rangle^{A_O} \langle i|^{A_I}.
\end{equation}

We can straightforwardly extend the "pure" definition of the Choi isomorphism to get a "mixed" version. Let $\mathcal{M}: \mathcal{L}(A_I) \to \mathcal{L}(A_O) $ be a linear map. It's Choi operator is defined as
\begin{equation}
M^{A_I A_O} = \sum_{ij} |i \rangle \langle j |^{A_I} \otimes \mathcal{M}(|i \rangle\langle j|)^{A_O},
\end{equation}
which is a positive operator if and only if $\mathcal{M}$ is completely-positive (CP)~\cite{Choi1975}. One may check that the isomorphism can be inverted by using the formula
\begin{equation}
\mathcal{M}(\rho) = \tr_{A_I} \left(M^{A_I A_O} \cdot \rho^{T_{A_I}} \otimes \id^{A_O} \right).
\end{equation}
The above equation can be used to show that $\mathcal{M}$ is trace-preserving iff $\tr_{A_O} M^{A_I A_O} = \id^{A_I}$. 

We also collect here the following identity, allowing to express the product of matrices in the Choi representation 
\begin{proposition}
\label{prop:mult_ket}
Let $P, A_I, A_O, F$ be isomorphic finite dimensional Hilbert spaces, and let $V_1 : P \to A_I$, $V_2: A_I \to A_O$, $V_3: A_O \to F$ be linear maps. Then
\begin{equation}
|V_3 V_2 V_1 \rrangle^{PF} = \llangle V_2^*|^{A_I A_O} |V_1 \rrangle^{P A_I} |V_3 \rrangle^{A_O F}.
\end{equation}
\end{proposition}
\begin{proof}
\begin{align}
\llangle V_2^*|^{A_I A_O} |V_1 \rrangle^{P A_I} |V_3 \rrangle^{A_O F} &= \sum_{ijk} \left(\langle k |^{A_I}\otimes (\langle k| V_2^T)^{A_O}  \right) \cdot \left((V_1^T |j \rangle)^P \otimes  |j \rangle^{A_I} \otimes |i \rangle^{A_O} \otimes (V_3|i\rangle)^{F} \right)\\
&= \sum_{ij}( \langle j| V_2^T |i \rangle) \cdot (V_1^T |j \rangle)^P (V_3|i\rangle)^{F} \\
& = \sum_{ij} (V_1^T |j \rangle)^P \otimes  (V_3 |i\rangle \langle i| V_2 |j \rangle)^F \\
&= \sum_j (V_1^T |j \rangle)^P \otimes  (V_3 V_2 |j \rangle)^F \\
&= \sum_j |j \rangle^P \otimes (V_3 V_2 V_1|j \rangle)^F = |V_3 V_2 V_1 \rrangle^{PF}.
\end{align}
\qed
\end{proof}

\section{Generalisation of Marcus' Theorem}
\label{app:marcus}
In this section we recall Marcus' theorem~\cite{Marcus1959}, and and its generalisation due to Cheung and Li~\cite{Cheung2003}. Let $\mathcal{H}_1, \mathcal{H}_2$ be Hilbert spaces, and let $f: \mathcal{L}(\mathcal{H}_1) \to \mathcal{L}(\mathcal{H}_2)$ be a map. We say that $f$ is \textit{unitarity preserving} if $f(U)$ is unitary for all unitaries $U \in \mathcal{L}(\mathcal{H}_1)$. In what follows Hilbert spaces are always finite-dimensional.
\begin{theorem}
\label{thm:marcus}
Let $\mathcal{H}$ be a finite-dimensional Hilbert space, and let $f: \mathcal{L}(\mathcal{H}) \to\mathcal{L}(\mathcal{H}) $ be a unitarity preserving linear map. Then either
\begin{align}
f(U) &= A U B\\
\mathrm{or} \, \, f(U) &= A U^T B,
\end{align}
where $A, B$ are constant unitary matrices, and where $T$ is the transpose in the computational basis.
\end{theorem}

In virtue of the Choi isomorphism, this theorem is equivalent to the fact that the only channels on a system of two qudits that preserve the set of maximally entangled states are products of local unitaries and swap~\cite{Brylinski2002, Hulpke2006}. 

We now state the generalisation of Marcus' theorem to the case where $f$ sends $d_1 \times d_1$ matrices to $d_2 \times d_2$ matrices, with $d_2$ an integer multiple of $d_1$.~\cite{Cheung2003}
\begin{theorem}(Cheung and Li~\cite{Cheung2003})
\label{thm:cheung}
Let $\mathcal{H}_1, \mathcal{H}_2$ be Hilbert spaces with dimensions $d_1$ and $d_2$, respectively. Let $f: \mathcal{L}(\mathcal{H}_1) \to \mathcal{L}(\mathcal{H}_{2})$ be a unitarity preserving linear map. Then
\begin{equation}
f(U) = S \left( (U \otimes \id_s) \oplus ( U^T \otimes \id_r) \right)T,
\end{equation}
where $d_2 = d_1(s + r)$ and $S, T \in \mathcal{L}(\mathcal{H}_2)$ are constant unitary matrices.
\end{theorem}

Note that if $f$ is a unitary preserving linear map as above, we can usually also assume that it is unital i.e. $f(\id) = \id$ without loss of generality, because if $f$ is a general unitary preserving linear map, then $f'(U) := f(\id)^\dagger f(U)$ is unital.

%

\section{The generalised Marcus theorem in the Choi representation}
\label{app:marcus_choi}

Let $A$ and $B$ be Hilbert spaces of dimensions $d_A, d_B$, respectively. Let $f: \LL(A) \to \LL(B)$ be a unitary preserving linear map. Then Theorem~\ref{thm:cheung} amounts to the following: There exists a direct sum decomposition of $B = (A^1 \otimes C) \oplus (A^2 \otimes D)$, where $A^1, A^2$ are copies of $A$ and where $C, D$ are (potentially zero-dimensional) Hilbert spaces whose dimension satisfy $(d_C + d_D) d_A = d_B$, and such that
\begin{equation}
|f(U) \rrangle^{B_I B_O} = (S^{B_O} \otimes (T^T)^{B_I} ) F |U \rrangle^{A_I A_O},
\end{equation}
where $S, T$ are constant unitaries and where
\begin{equation}
F = \sum_{i, j = 0}^{d_A - 1} \left( |i \rangle^{A_I^1} |j \rangle^{A_O^1} \langle i|^{A_I} \langle j |^{A_O} \otimes |\id\rrangle^{C_I C_O} \oplus   |i \rangle^{A_I^2} |j \rangle^{A_O^2} \langle j|^{A_I} \langle i |^{A_O} \otimes |\id\rrangle^{D_I D_O} \right).
\end{equation}
To prove this claim it suffices to verify that $F |U \rrangle = |(U^{A^1} \otimes \id^{C} ) \oplus ((U^T)^{A^2} \otimes \id^{D}) \rrangle$, which is straightforward. 

We first consider the simpler case where $S = T = \id$. We apply the Choi isomorphism to $F$ to obtain
\begin{equation}
|F \rrangle = |\id \rrangle^{A_I A_I^1} |\id \rrangle^{A_O^1 A_O}  |\id \rrangle^{C_I C_O} \oplus |\id \rrangle^{A_I A_O^2} |\id \rrangle^{A_O A_I^1} |\id \rrangle^{D_I D_O}.
\end{equation}
One can obtain information about the dimensions of $C$ and $D$ by taking some partial traces on $|F \rrangle$, as follows:
\begin{align}
\tr_{B_I} |F \rrangle \llangle F| &= \id^{A_I C_O} \otimes |\id \rrangle \llangle \id|^{A_O A_O^1} + \id^{A_O D_O} \otimes |\id \rrangle \llangle \id|^{A_I A_O^2} \\
\tr_{A_I B_I} |F \rrangle \llangle F| &= d_A |\id \rrangle \llangle \id|^{A_O A_O^1} \otimes \id^{C_O} + \id^{A_O A_O^2 D_O}
\end{align}
so this reduced state has eigenvalues $d_A^2$ and $1$, with respective multiplicities $d_C$ and $d_A^2d_D$.

For the general case, if we define $Q =( S^{B_O} \otimes (T^T)^{B_I} )F$, so that $|f(U) \rrangle = Q |U\rrangle$, we have
$|Q \rrangle = S^{B_O} \otimes (T^{T})^{B_I} |F \rrangle$, which implies that
\begin{equation}
\tr_{A_I B_I} |Q \rrangle \llangle Q| = S^{B_O} \left( \tr_{A_I B_I} |F \rrangle \llangle F| \right) (S^\dagger)^{B_O}
\end{equation}
has the same eigenvalues as $\tr_{A_I B_I} |F \rrangle \llangle F|$.

\section{The time reversal of a pure process}
\label{app:time_rev}

Let $|w \rangle$ be a pure process vector whose parties have equal input and output Hilbert space dimensions. Taking the complex conjugate and swapping inputs and outputs yield a valid process
\begin{equation}
|w_{r}\rangle^{P A_I A_O B_I B_O F} := (SWAP_{PF} \otimes SWAP_{A_I A_O} \otimes SWAP_{B_I B_O}) |w^*\rangle^{P A_I A_O B_I B_O F},
\end{equation}
which we call the time-reversal of $|w\rangle$. We now show that $|w_r \rangle$ is a valid process.
\begin{align}
|w_r\rangle &= \frac{1}{d_A d_B} \sum_{i,j} |\mathcal{G}(U_i,V_j)^* \rrangle^{FP} |U_i \rrangle^{A_O A_I} |V_j \rrangle^{B_O B_I} \\
&= \frac{1}{d_A d_B}\sum_{i,j} |\mathcal{G}(U_i,V_j)^\dagger \rrangle^{PF} |U_i^T \rrangle^{A_I A_O} |V_j^T \rrangle^{B_I B_O} \\
&=\frac{1}{d_A d_B} \sum_{i,j} |\mathcal{G}(U_i^\dagger ,V_j^\dagger)^\dagger \rrangle^{PF} |U_i^* \rrangle^{A_I A_O} |V_j^* \rrangle^{B_I B_O} \\
& = \frac{1}{d_A d_B}\sum_{i,j} |\mathcal{G}_{r}(U_i,V_i) \rrangle^{PF} |U_i^* \rrangle^{A_I A_O} |V_j^* \rrangle^{B_I B_O},
\end{align}
where in the third line we made a basis change $U_i \mapsto U_i^\dagger, V_j \mapsto V_j^\dagger$. The last equation shows that the reversed process is equivalently defined by the map
\begin{equation}
\mathcal{G}_{r}(U, V) := \mathcal{G}(U^\dagger ,V^\dagger)^\dagger.
\end{equation}
The map $\mathcal{G}_{r}$ admits a decomposition into causal frames:
\begin{align}
\mathcal{G}_{r}(U_A, U_B) &= \left( \Phi_A(U_B^\dagger) (U_A^\dagger \otimes \id^{E_A} ) \Pi_A(U_B^\dagger) \right)^\dagger \\
&= \Pi_A(U_B^\dagger)^\dagger (U_A \otimes \id^{E_A} ) \Phi_A(U_B^\dagger)^\dagger.
\end{align}
The above equation shows that Alice's causal reference frame is given by $\Phi_A^{r}(U_B) = \Pi_A(U_B^\dagger)^\dagger$, $\Pi_A^r(U_B) = \Phi_A(U_B^\dagger)^\dagger$, and similarly for Bob. Theorem~\ref{thm:frames_implies_pure} then implies that $|w_r \rangle$ is a valid pure process.

As a simple example, and as justification for calling this operation ``time-reversal'', consider the single partite process $\mathcal{G}(U) = A U B$, where $A,B$ are fixed unitaries. Then it's time-reverse is $\mathcal{G}_{r}(U) = B^\dagger U A^\dagger$.

\end{document}